\documentclass[12pt]{article}
\usepackage[T1]{fontenc}
\usepackage{amsmath}
\usepackage{float}
\usepackage{amsfonts}
\usepackage{graphicx}
\usepackage[mathscr]{euscript} 
\usepackage{rotating}
\usepackage{pdflscape}
\usepackage[hmargin=2.5cm,vmargin=3cm]{geometry} %Seems to fit ECMA margins
\usepackage{booktabs,caption}
\usepackage{natbib}
\usepackage{setspace}
\usepackage{bbm}
\usepackage{amsthm}
\usepackage{enumerate}
\usepackage{threeparttable}
\usepackage{verbatim}
\usepackage{caption}
\usepackage{url}
\usepackage{amssymb}
\usepackage{subcaption,array}
\usepackage{adjustbox}
\usepackage[titletoc,title]{appendix}
\usepackage{multirow}
\usepackage{threeparttable}
\usepackage{siunitx}

\usepackage{newtxtext,newtxmath}
\newcommand\fnote[1]{\captionsetup{font=small}\caption*{#1}}

\usepackage[hidelinks,breaklinks]{hyperref}
\hypersetup{colorlinks=true,allcolors=blue}

\usepackage{framed}

\newtheorem{prop}{Proposition}

\theoremstyle{definition}
\newtheorem{assumption}{Assumption}
\newtheorem{define}{Definition}

\usepackage{calligra}
\usepackage[T1]{fontenc}

\usepackage{todonotes}

\usepackage[scr=boondoxo]{mathalfa}  % Use 'boondox' script font
\newcommand{\symms}{\mathscr{s}}

\usepackage{longtable}

\begin{document}

\title{
    \textbf{Nested Pseudo-GMM Estimation \\ 
    of Demand for Differentiated Products}\thanks{We thank Zhentong Lu for their helpful comments. The data used in this study were obtained through a request submitted under the \textit{Canadian Access to Information Act}. We thank the \textit{Liquor Control Board of Ontario} (LCBO) for providing access to the data.}
    }
    
\author{
    Victor Aguirregabiria \footnote{Department of Economics, University of Toronto. 150 St. George Street, Toronto, ON, M5S 3G7, Canada, \href{mailto.victor.aguirregabiria@utoronto.ca}{victor.aguirregabiria@utoronto.ca}.}  \\  \emph{University of Toronto, CEPR} \and
    Hui Liu
    \footnote{Ma Yinchu School of Economics, Tianjin University, Tianjin, China, \href{mailto: liuhui33@tju.edu.cn}{liuhui33@tju.edu.cn}.} 
    \\ \emph{Tianjin University}
    \and
    Yao Luo\footnote{Department of Economics, University of Toronto. 150 St. George Street, Toronto, ON, M5S 3G7, Canada, \href{mailto: yao.luo@utoronto.ca}{yao.luo@utoronto.ca}.} 
    \\ \emph{University of Toronto}
    }

\date{\today}

\maketitle

\thispagestyle{empty}

\begin{abstract} 
    We propose a fast algorithm for computing the GMM estimator in the BLP demand model (\citeauthor{berry_levinshon_1995}, \citeyear{berry_levinshon_1995}). Inspired by nested pseudo-likelihood methods for dynamic discrete choice models, our approach avoids repeatedly solving the inverse demand system by swapping the order of the GMM optimization and the fixed-point computation. We show that, by fixing consumer-level outside-option probabilities, BLP’s market-share–mean-utility inversion becomes closed-form and, crucially, separable across products, yielding a nested pseudo-GMM algorithm with analytic gradients. The resulting estimator scales dramatically better with the number of products and is naturally suited for parallel and multithreaded implementation. In the inner loop, outside-option probabilities are treated as fixed objects while a pseudo-GMM criterion is minimized with respect to the structural parameters, substantially reducing computational cost. Monte Carlo simulations and an empirical application show that our method is significantly faster than the fastest existing alternatives, with efficiency gains that grow more than proportionally in the number of products.

\vspace{0.4cm}
\noindent
\textbf{Keywords:} Random Coefficients Logit; Sufficient Statistics; Market Share Inversion; Newton-Kantorovich Iteration; Asymptotic Properties; LCBO

\vspace{0.4cm}
\noindent\textbf{JEL codes:} C23, C25, C51, C61, D12, L11.
\end{abstract}

\newpage
\setcounter{page}{1}

\setstretch{1.6}

%%%%%%%%%%%%%%%%%%%%%%

\section{Introduction}

The \cite{berry_levinshon_1995} (hereafter BLP) random coefficients discrete-choice model has become a workhorse of demand modeling in the empirical literature, as it accommodates rich substitution patterns between products and avoids the \textit{independence of irrelevant alternatives} property inherent in standard multinomial logit models. The original estimation approach relies on a nested fixed-point algorithm, which combines the minimization of a generalized method of moments (GMM) criterion function in the outer loop with a fixed-point iteration in the inner loop to invert the demand system and recover mean utilities. This method has been widely adopted in empirical research, including \citet{nevo2000practitioner}, \citet{petrin2002quantifying}, \citet{bayer2007unified}, and \citet{fan2013ownership}, among many others.\footnote{See \citet{ackerberg2007econometric} for a comprehensive survey, and \citet{conlon2020best} for a discussion of best practices in implementing this algorithm.}

It is well known that the nested fixed-point algorithm is computationally inefficient for estimating the GMM parameters in the BLP model. While the outer-loop minimization steps are relatively inexpensive, the inner-loop fixed-point iterations are computationally intensive, as they involve inversion of a high-dimensional demand system. The primary source of inefficiency is that the algorithm fully solves the costly fixed-point problem for every candidate parameter vector, even when those guesses are far from the final estimate.

Several studies have proposed alternative methods to reduce the computational burden of estimating the BLP model. \citet{dube2012improving} introduce a mathematical program with equilibrium constraints (MPEC) approach, which minimizes the GMM objective function subject to the market share equations treated as constraints. By jointly optimizing the Lagrangian over both the structural parameters and the mean utilities, the MPEC method requires solving the market share inversion only once—at the optimum. \citeauthor{ABLP2015} (\citeyear{ABLP2015}) propose a method that uses a linear approximation to analytically invert the market share equations, leading to a pseudo-GMM minimization problem over the model parameters. We refer to this approach as the ABLP estimator.

Inspired by the Nested Pseudo Likelihood (NPL) estimator of \citet{aguirregabiria_mira_2002,aguirregabiria_mira_2007} in the context of dynamic discrete choice models, we propose a Nested Pseudo-GMM (NP-GMM) method for estimating the BLP demand model. Our approach incorporates several features that substantially reduce the computational burden of GMM estimation. NP-GMM is a computationally motivated estimator that targets the same structural parameters under the same identification conditions as GMM. It trades a small asymptotic efficiency loss for large speed gains.

First, similar to \citet{ABLP2015} and in contrast to the standard nested fixed-point algorithm, our method swaps the order of the fixed-point iteration and the GMM optimization. This avoids the repeated—and computationally intensive—inversion of the demand system, at the cost of performing more, but less expensive, iterations in the minimization of a pseudo-GMM criterion. Second, in the inner loop, we treat the vector of individual outside-option choice probabilities as incidental objects (i.e., as \textit{sufficient statistics} for competitive effects). Conditional on these probabilities, mean utilities can be recovered product by product using a closed-form expression: that is, the market-share–mean-utility inversion is both closed-form and, crucially, separable across products. Third, while the gradient of the full GMM objective function lacks an analytical expression, the gradient of our pseudo-GMM criterion function is simple and explicit. This facilitates efficient computation of both the objective function and its derivatives and makes our algorithm particularly well-suited for parallel processing and multithreading. Together, these innovations lead to a significantly faster estimation procedure. Importantly, while the defining feature in 
\citet{ABLP2015}'s ABLP method is a Taylor approximation / linearization of the fixed point mapping, our method instead uses an exact identity, not an approximation. 

Our NP-GMM method outperforms alternative \textit{Nested Pseudo} (NP) algorithms that more directly adapt the insights from \citet{aguirregabiria_mira_2002, aguirregabiria_mira_2007} to the BLP framework. Our approach extends these alternative methods in several important ways. First, rather than fixing the mean utilities as in a direct adaptation, we fix individual choice probabilities for the outside option to implement the pseudo-GMM step in the inner loop. This yields simple, closed-form expressions that are well-suited for parallelization and multithreading. Second, we show that our estimation problem can be reformulated as one in which the set of parameters is augmented with an auxiliary set, which is held fixed during the pseudo-GMM step. This recasts the problem into a structure similar to that of the original Nested Pseudo estimators. 

We establish consistency and asymptotic normality of our NP-GMM estimator, providing a new asymptotic framework for NP estimation in high- and infinite-dimensional settings, where existing proofs do not apply. In many applications of NP methods—such as those in \citet{lin2024endogeneity} and \citet{dearing2025efficient}—it is standard to assume a finite number of states, which underpins the asymptotic arguments in \citet{aguirregabiria_mira_2002, aguirregabiria_mira_2007}. However, the BLP model differs fundamentally from standard dynamic discrete choice models: the dimensionality of the mean utility vector grows with the sample size and tends to infinity. This renders conventional asymptotic arguments for NP estimators inapplicable. In contrast, our reformulation offers a novel framework for establishing the statistical properties of NP-GMM methods in such settings.

We conduct a series of Monte Carlo experiments to evaluate the computational efficiency and finite-sample performance of the proposed method. Our comparison focuses on the ABLP estimator, which is the most closely related approach and, to date, the fastest available method for estimating random-coefficients demand models using aggregate data. As shown by \citet{ABLP2015}, the ABLP estimator delivers substantial computational gains relative to traditional nested fixed-point and MPEC methods. We then illustrate the implementation of our methodology using wine demand data from the Liquor Control Board of Ontario (LCBO), demonstrating the scalability and feasibility of estimating demand models with a large number of products.

A key advantage of our method lies in its compatibility with multithreading. While our algorithm benefits significantly from parallelization when computing the GMM objective function and its gradients, the ABLP estimator performs poorly under multithreading, particularly when evaluating gradients—a major component of its overall computational cost. As a result, ABLP gains little from parallelization. When comparing both methods using their respective optimal number of threads, we find that our algorithm evaluates the GMM objective function and gradients approximately five times faster than ABLP. Even for a moderate number of products, our proposed method is approximately twice as fast as the fastest currently available estimator, and its computational advantage grows more than proportionally with the number of products in the demand system. 

Our paper also relates to a broader literature that develops computationally efficient estimators for variants of the BLP model. Assuming a discretization of consumer tastes and a reduced-form pricing function, \citet{fox2011simple} propose a linear regression estimator subject to linear inequality constraints. \citet{kalouptsidi2012market} also considers discrete consumer types and leverages a duality between market shares and consumer heterogeneity. More recently, \citet{lu2023semi} introduce a semi-nonparametric approach and propose a two-step estimator based on sieve GMM. See also \citet{wang2023sieve} and \citet{compiani2022market} for related sieve-based methods for estimating nonparametric distributions of random coefficients in the BLP framework. In addition, several papers examine the numerical performance and implementation best practices of the standard BLP estimator and its alternatives; see, for example, \citet{knittel2014estimation} and \citet{conlon2020best}.

The rest of the paper is organized as follows. Section \ref{sec:model} introduces the BLP demand model, the GMM estimator, and three estimation algorithms: the Nested Fixed Point method, the ABLP approach, and our proposed NP-GMM algorithm. Section \ref{sec:monte} presents Monte Carlo evidence comparing the performance of the ABLP and NP-GMM estimators. Section \ref{sec:lcbo} illustrates the implementation of the proposed algorithm and demonstrates its computational performance using wine demand data from the LCBO. Section \ref{sec:conclusion} summarizes and concludes. The code implementing the different estimators and the Monte Carlo experiments is available on GitHub at
\url{https://github.com/luo8yao/np-gmm-blp}.

\section{Model, estimators, and algorithms \label{sec:model}} 

\subsection{BLP Model}

We consider a set of markets indexed by \( t \in \{1,2,\ldots,T\} \). Each market is populated by a mass \( M_t \) of consumers and offers a set of products \( \mathcal{J} = \{1,2,\ldots,J\} \). In each market, consumers choose one product \( j \in \mathcal{J} \) or opt not to purchase, which we denote by \( j = 0 \). Product \( j \) in market \( t \) is characterized by observed attributes \( \boldsymbol{x}_{jt} \) and an unobserved quality component \( \xi_{jt} \). The vector \( \boldsymbol{x}_{jt} \) contains \( K \) product characteristics, including price, and we let \( \boldsymbol{x}_t \) denote the collection \( (\boldsymbol{x}_{1t},\ldots,\boldsymbol{x}_{Jt}) \). The scalar \( \xi_{jt} \) captures market- and product-specific utility components observed by consumers but unobserved by the researcher; we collect these terms in the vector \( \boldsymbol{\xi}_t = (\xi_{1t},\ldots,\xi_{Jt}) \). Consumer $i$ in market $t$ obtains the utility from purchasing product $j$:
\begin{equation}
    u_{ijt} = \boldsymbol{x}_{jt}^{\prime} \, \boldsymbol{\beta}_{it} + \xi_{jt} + \epsilon_{ijt}.
\end{equation}
The utility of the outside good, the "no-purchase" option, is $u_{i0t} = \epsilon_{i0t}$. The parameter vector $\boldsymbol{\beta}_{it} = (\beta^{1}_{it}, \beta^{2}_{it}, ..., \beta^{K}_{it})^{\prime}$ contains consumer $i$'s (in market $t$) tastes for the $K$ characteristics; $\epsilon_{ijt}$ is an additional idiosyncratic shock that follows the Type I extreme value distribution. 

For $k \in \{1,2,...,K\}$, each random coefficient $\beta^{k}_{it}$ has the following structure:
\begin{equation}
    \beta^k_{it} = \beta^k + 
    \sigma^{k} \, \nu^{k}_{it} + 
    \boldsymbol{d}_{it}^{\prime} \, \boldsymbol{\pi}^{k},
\end{equation}
where $\beta^k$ and $\sigma^{k}$ are fixed parameters; $(\nu_{it}^{1}, \ldots, \nu_{it}^{K})$ are independent random draws from the standard normal distribution; $\boldsymbol{d}_{it}$ is a vector of observable demographics of consumer $i$ in market $t$, which may be measured at the consumer or market level; and $\boldsymbol{\pi}^{k}$ is a vector of parameters.

Let $\boldsymbol{\theta}$ be the vector with all the model parameters. It is convenient to distinguish two groups of parameters: $\boldsymbol{\theta}^{\prime} = \left( \boldsymbol{\beta}^{\prime}, \, \boldsymbol{\sigma}^{\prime} \right)$, with 
$\boldsymbol{\beta}^{\prime} = \left( \beta^{1}, \, ..., \,  \beta^{K} \right)$ and $\boldsymbol{\sigma}^{\prime} = \left( \sigma^{k}, \,  \boldsymbol{\pi}^{k}: \, k=1,2, ..., K \right)$. Let $\delta_{jt}$ be the average utility of product $j$ in market $t$: that is, $\delta_{jt} = \boldsymbol{x}_{jt}^{\prime} \, \boldsymbol{\beta} + \xi_{jt}$. Define the vector of average utilities for all the products in market $t$ as $\boldsymbol{\delta}_{t}$. The market share of product $j$ in market $t$ is represented as $s_{jt}$. We use function $\symms_{j}(\boldsymbol{\delta}_{t}, \boldsymbol{x}_t, \boldsymbol{\sigma})$ to represent the prediction of the model about market share $s_{jt}$. That is:
\begin{equation}
    s_{jt} = 
    \symms_{j}(\boldsymbol{\delta}_{t}, \boldsymbol{x}_t, 
    \boldsymbol{\sigma}) 
    \, = \,
    \int_{\boldsymbol{\nu}_{it}, \boldsymbol{d}_{it}} 
    \frac{ \exp \{ 
        \delta_{jt} + 
        \sum_{k=1}^{K} x^{k}_{jt} \,
        [\sigma^{k} \, \nu^{k}_{it} + 
        \boldsymbol{d}_{it}^{\prime} \, \boldsymbol{\pi}^{k} ]
        \}
    }
    { 1 + \sum_{j^{\prime}=1}^J 
    \exp \{ 
        \delta_{j^{\prime} t} + 
        \sum_{k=1}^{K} x^{k}_{j^{\prime} t} \,
        [ \sigma^{k} \, \nu^{k}_{it} + 
        \boldsymbol{d}_{it}^{\prime} \, \boldsymbol{\pi}^{k} 
        ]
        \}
    } \; dF_{t}(\boldsymbol{\nu}_{it}, \boldsymbol{d}_{it})
\end{equation}
where $F_{t}(\cdot)$ represents the distribution function of $(\boldsymbol{\nu}_{it}, \boldsymbol{d}_{it})$ in market $t$. This integral is typically evaluated via Monte Carlo simulation, using $N$ draws from the distribution $F_t$. Throughout the paper, we assume that market shares are computed using a finite number of such draws.\footnote{This can be interpreted either as assuming a finite population or that $N$ is sufficiently large for the simulation error to be negligible—an assumption that is nearly universal in this literature.} Accordingly, we treat the set of draws $\{ \boldsymbol{\nu}_{it}, \boldsymbol{d}_{it}: i = 1, 2, \dots, N \}$ as fixed and known to the researcher, and we use the following expression for the model-predicted market shares:
\begin{equation}
    s_{jt} \; = \; 
    \symms_{j}(\boldsymbol{\delta}_{t}, \boldsymbol{x}_t, 
    \boldsymbol{\sigma}) 
    \; = \;
    \displaystyle
    \frac{1}{N}
    \sum_{i=1}^{N}
    \frac{ \exp \{ 
        \delta_{jt} + 
        \sum_{k=1}^{K} x^{k}_{jt} \,
        [\sigma^{k} \, \nu^{k}_{it} + 
        \boldsymbol{d}_{it}^{\prime} \, \boldsymbol{\pi}^{k} ]
        \}
    }
    { 1 + \sum_{j^{\prime}=1}^J 
    \exp \{ 
        \delta_{j^{\prime} t} + 
        \sum_{k=1}^{K} x^{k}_{j^{\prime} t} \,
        [ \sigma^{k} \, \nu^{k}_{it} + 
        \boldsymbol{d}_{it}^{\prime} \, \boldsymbol{\pi}^{k} 
        ]
        \}
    }
\label{eq_model_mshares}
\end{equation}

Define the vector of market shares $\boldsymbol{s}_{t} \equiv \left( s_{1t}, \, ..., \, s_{Jt} \right)$ and the corresponding 
vector-valued mapping 
$\boldsymbol{\symms}(\boldsymbol{\delta}_{t}, \boldsymbol{x}_t, \boldsymbol{\sigma}) \equiv 
\left( \symms_{1}(\boldsymbol{\delta}_{t}, \boldsymbol{x}_t, \boldsymbol{\sigma}), \, ..., \,
\symms_{J}(\boldsymbol{\delta}_{t}, \boldsymbol{x}_t, \boldsymbol{\sigma}) \right)$. The demand system, in vector form, is:
\begin{equation} \label{eq:demand}
    \boldsymbol{s}_{t} \; = \;
    \boldsymbol{\symms}(\boldsymbol{\delta}_{t}, \boldsymbol{x}_t, \boldsymbol{\sigma})
\end{equation}
By \textit{Berry's Inversion Theorem} (\citeauthor{berry1994}, \citeyear{berry1994}), this mapping is invertible in $\boldsymbol{\delta}_{t}$. For any value of $(\boldsymbol{x}_t, \boldsymbol{\sigma})$, there exists an inverse function $\boldsymbol{\symms}^{-1}(\boldsymbol{s}_{t}, \boldsymbol{x}_t, \boldsymbol{\sigma})$ such that $\boldsymbol{\delta}_{t} = \boldsymbol{\symms}^{-1}(\boldsymbol{s}_{t}, \boldsymbol{x}_t, \boldsymbol{\sigma})$, or at the product level, 
\begin{equation}
    \delta_{jt} \; = \;
    \symms^{-1}_{j}(\boldsymbol{s}_{t}, \boldsymbol{x}_t, \boldsymbol{\sigma})
\end{equation}
While the inverse function has a closed-form expression in the standard logit and nested logit models, this is not the case in the random coefficients logit model. In this case, the inversion must be computed numerically, which is computationally intensive—particularly in models with a large number of products.

Based on the inversion of the demand system, we can write this demand model as a partially linear regression model:
\begin{equation}
    \symms^{-1}_{j}(\boldsymbol{s}_{t}, \boldsymbol{x}_t, \boldsymbol{\sigma})
    \; = \;
    \boldsymbol{x}_{jt}^{\prime} \, \boldsymbol{\beta} + \xi_{jt} 
\end{equation}
Price is included in \( \boldsymbol{x}_{jt} \) and is potentially correlated with \( \xi_{jt} \), so consistent estimation requires instrumental variables. The econometric model is completed by specifying the validity and identification power of instrumental variables. Specifically, we assume the existence of a vector of $q$ observable variables $\boldsymbol{z}_{jt}$ that satisfies the moment conditions $\mathbb{E}\left( \boldsymbol{z}_{jt} \, \xi_{jt} \right) = \boldsymbol{0}$. These conditions identify the parameter vector $\boldsymbol{\theta}$. A common choice for the instruments includes observable product characteristics other than the price of product $j$, along with the same characteristics for products that compete with product $j$.

\subsection{The GMM estimator and the Nested Fixed Point algorithm \label{sec:NFP}}

Let $m(\boldsymbol{\theta})$ be the vector of $q$ sample moments that correspond to the population moments $\mathbb{E}\left( \boldsymbol{z}_{jt} \, \xi_{jt} \right)$.
\begin{equation}
    \displaystyle 
    m(\boldsymbol{\theta}) \; = \;
    \frac{1}{JT} 
    \sum_{t=1}^{T} \sum_{j=1}^{J}
    \boldsymbol{z}_{jt} \, 
    \left[ 
        \symms^{-1}_{j}(\boldsymbol{s}_{t}, \boldsymbol{x}_t, \boldsymbol{\sigma}) -
        \boldsymbol{x}_{jt}^{\prime} \, \boldsymbol{\beta} 
    \right]
\end{equation}
The GMM estimator of $\boldsymbol{\theta}$ is defined as:
\begin{equation}
    \widehat{\boldsymbol{\theta}} \; = \;
    \arg \min_{\boldsymbol{\theta} \in \Theta}
    \; \;
    G(\boldsymbol{\theta}) 
    \; = \;
    m(\boldsymbol{\theta})^{\prime} \; \boldsymbol{W} \;
    m(\boldsymbol{\theta})
\end{equation}
where $\boldsymbol{W}$ is a positive definite $q \times q$ weighting matrix. 

The GMM estimator $\widehat{\boldsymbol{\theta}}$ satisfies the first-order conditions: $\frac{\partial  G\widehat{\boldsymbol{\theta}})}{\partial \boldsymbol{\theta}} = 0$. 
\cite{berry_levinshon_1995} propose the Nested Fixed-Point (NFXP) algorithm that computes this estimator by applying Newton’s method to solve this system of nonlinear equations \eqref{eq:demand}. Specifically, it performs the following iterative procedure until convergence. At iteration \(n+1 \geq 1\):  
\begin{equation}
    \boldsymbol{\theta}^{n+1}
    \, = \,
    \boldsymbol{\theta}^{n} - 
    \left( 
        \frac{\partial^{2} \, G(\boldsymbol{\theta}^{n})}
        {\partial \, \boldsymbol{\theta} \, 
        \boldsymbol{\theta}^{\prime}}
    \right)^{-1}
    \; 
    \frac{\partial \, G(\boldsymbol{\theta}^{n})}
        {\partial \, \boldsymbol{\theta}}
\end{equation}

At every Newton iteration, and for every market $t$ in the data, the NFXP algorithm fully solves the inversion problem that defines the vector of average utilities $\boldsymbol{\delta}_{t}$ as an implicit function of $(\boldsymbol{s}_{t}, \boldsymbol{x}_t, \boldsymbol{\sigma})$. Specifically, for a fixed value of $\boldsymbol{\sigma}$ (the one from Newton's iteration $n$), it applies the following fixed point iterations until convergence. At iteration \(\tau +1 \geq 1\):  
\begin{equation}
    \boldsymbol{\delta}_{t}^{\tau+1}
    \; = \;
    \boldsymbol{\delta}_{t}^{\tau} + 
    \ln\left( \boldsymbol{s}_{t} \right) - 
    \ln\left( 
        \boldsymbol{\symms} 
        (\boldsymbol{\delta}_{t}^{\tau}, \boldsymbol{x}_t, \boldsymbol{\sigma})
    \right).
\end{equation}
Upon convergence of the fixed-point algorithm, the solution $\boldsymbol{\delta}_t$ represents the evaluation of the inverse function $\boldsymbol{\symms}^{-1}(\boldsymbol{s}_t, \boldsymbol{x}_t, \boldsymbol{\sigma})$. This result can then be used to compute the criterion function $G(\boldsymbol{\theta}^{n})$ and its derivatives, enabling a new Newton iteration that updates the parameter vector to a new value $\boldsymbol{\theta}^{n+1}$.

In empirically common situations --especially when the outside good share is small-- the modulus of this contraction mapping is often close to one, which means the fixed-point algorithm may require many iterations to converge 
(\citeauthor{dube2012improving}, \citeyear{dube2012improving},  \citeauthor{conlon2020best}, \citeyear{conlon2020best}). The NFXP algorithm must solve $T$ such fixed-point problems at each trial value of $\boldsymbol{\theta}$ during the Newton iterations. These characteristics make the NFXP algorithm computationally intensive.

\subsection{Aguirregabiria-Mira Nested Pseudo-GMM estimator}

Before introducing our estimator, we begin by describing a direct application of the Nested Pseudo-Likelihood or GMM method from \citet{aguirregabiria_mira_2002, aguirregabiria_mira_2007} to the BLP model. This section serves two purposes. First, it introduces the core ideas shared by a broad class of Nested Pseudo-GMM (NP-GMM) estimators. Second, it demonstrates that a straightforward, off-the-shelf application of the standard NP-GMM approach—commonly used in dynamic discrete choice models—does not work in the context of the BLP model. These limitations motivate the modifications to this method that we present in Section \ref{sec_our_NPGMM}.

Let $\boldsymbol{\delta} = (\boldsymbol{\delta}_{t}: t=1, 2, ..., T)$ be the vector of mean utilities across all the markets and products in the sample. The method treats these mean utilities as incidental parameters to be estimated together with the structural parameters in $\boldsymbol{\theta}$. Specifically, we define an extended or pseudo-GMM criterion function that includes $(\boldsymbol{\delta}, \boldsymbol{\theta})$ as arguments:
\begin{equation}
    \widetilde{Q}(\boldsymbol{\delta}, \boldsymbol{\theta}) 
    \; = \;
    \displaystyle 
    \left(
    \sum_{j,t}
    \boldsymbol{z}_{jt} \, 
    \left[ 
        \delta_{jt} -
        \boldsymbol{x}_{jt}^{\prime} \, \boldsymbol{\beta} 
    \right]
    \right)^{\prime} 
    \, \boldsymbol{W} \,
    \left(
    \sum_{j,t}
    \boldsymbol{z}_{jt} \, 
    \left[ 
        \delta_{jt} -
        \boldsymbol{x}_{jt}^{\prime} \, \boldsymbol{\beta} 
    \right]
    \right)
\label{eq_am_criterion}
\end{equation}

\begin{define} \label{def_am_npgmm}
    The \textbf{AM-NP-GMM estimator} is defined as the pair $(\widehat{\boldsymbol{\delta}}, \widehat{\boldsymbol{\theta}})$   that satisfies the following two conditions:
    \begin{itemize}
        \item AM-1. Given the vector of mean utilities $\widehat{\boldsymbol{\delta}}$, the parameter vector $\widehat{\boldsymbol{\theta}}$ minimizes the pseudo-GMM criterion function:
        \begin{equation}
            \widehat{\boldsymbol{\theta}} 
            \; = \;
            \arg \min_{\boldsymbol{\theta}} 
            \; \; 
            \widetilde{Q}(\widehat{\boldsymbol{\delta}}, \boldsymbol{\theta}) \; .
        \end{equation}

        \item AM-2. Given the structural parameters $\widehat{\boldsymbol{\theta}}$, the vector of mean utilities $\widehat{\boldsymbol{\delta}}_{t}$ solves, for each market $t$, the system of $J$ nonlinear equations implied by the demand system:
        \begin{equation}
            \boldsymbol{s}_{t}
            \; = \;
            \boldsymbol{\symms} 
            (\widehat{\boldsymbol{\delta}}_{t}, \, 
            \boldsymbol{x}_t, \,
            \widehat{\boldsymbol{\theta}}) 
            \qquad \text{for } 
            t = 1, \, 2, \, \dots, \, T.    
            \qquad \blacksquare
        \end{equation}
    \end{itemize}
\end{define}

A key feature concerns the algorithm used to compute the estimator $(\widehat{\boldsymbol{\delta}}, \widehat{\boldsymbol{\theta}})$ that satisfies conditions AM-1 and AM-2. This algorithm shares several characteristics with the NFXP algorithm: it is a nested procedure that alternates between minimizing the pseudo-GMM criterion with respect to $\boldsymbol{\theta}$ and performing iterations to solve for the mean utilities $\boldsymbol{\delta}_t$ implied by the demand system. However, unlike the NFXP algorithm, the roles of the inner and outer loops are reversed.

The algorithm proceeds as follows. It starts from an initial value for the vectors of means utilities: $\widehat{\boldsymbol{\delta}}^{0} = (\widehat{\boldsymbol{\delta}}^{0}_{t}: \, t \,= \,1, \, 2, \, ..., \,T)$. A natural way of initializing the mean utilities is using the ones that we would have if the model did not have random coefficients (i.e., $\boldsymbol{\sigma} = 0$) such that we had a standard logit model:
\begin{equation}
    \widehat{\delta}^{0}_{jt} \; = \;
    \ln(s_{jt}) - \ln(s_{0t})
\end{equation}
Then, at every (outer) iteration $n \geq 1$ we perform two tasks:
\begin{itemize}
    \item [1.] \textit{Pseudo-GMM estimation} of 
    $\boldsymbol{\theta}$, given $\widehat{\boldsymbol{\delta}}^{n-1}$:
    \begin{equation}
        \widehat{\boldsymbol{\theta}}^{n}
        \; = \;
        \arg \min_{\boldsymbol{\theta}} 
        \; \; 
        \widetilde{Q}
        \left(\widehat{\boldsymbol{\delta}}^{n-1}, \boldsymbol{\theta} \right) .
    \end{equation}
    \item [2.] \textit{One Newton-Kantorovich iteration} to update the vector of mean utilities:
    \begin{equation}
    \widehat{\boldsymbol{\delta}}^{n}_{t}
    \; = \;
    \widehat{\boldsymbol{\delta}}^{n-1}_{t} + 
    \left[ 
        \nabla_{\delta^{\prime}} \ln\boldsymbol{\symms}(\widehat{\boldsymbol{\delta}}^{n-1}_{t}, \widehat{\boldsymbol{\sigma}}^{n})
    \right]^{-1} \, 
    \left[ 
        \ln \boldsymbol{s}_{t} - 
        \ln \boldsymbol{\symms}(\widehat{\boldsymbol{\delta}}^{n-1}_{t}, \widehat{\boldsymbol{\sigma}}^{n})
    \right] 
    \qquad \text{for } 
        t = 1, \, 2, \, \dots, \, T. \label{eq_outer_npgmma_newton}
    \end{equation}
    where $\nabla_{\delta^{\prime}} \ln\boldsymbol{\symms}(\boldsymbol{\delta}, \boldsymbol{\sigma})$ denotes the $J \times J$ Jacobian matrix of market shares with respect to mean utilities. 
\end{itemize}

The main computational advantage of this algorithm over NFXP is that it avoids repeatedly solving the $T$ fixed-point problems. Instead, the NP-GMM algorithm solves these $T$ fixed-point problems only once -—at convergence. This efficiency comes at the cost of solving multiple pseudo-GMM minimization problems. However, these minimizations are computationally inexpensive because the pseudo-GMM criterion function is quadratic in $\boldsymbol{\theta}$ such that the pseudo-GMM estimator has a simple closed-form solution.

However, the approach described above suffers from a fundamental limitation. The pseudo-GMM criterion function in equation \eqref{eq_am_criterion} depends on the incidental parameters $\boldsymbol{\delta}$ and a subset of the structural parameters, $\boldsymbol{\beta}$, but not on the remaining structural parameters, $\boldsymbol{\sigma}$. As a result, these parameters cannot be identified or estimated within this framework.

In the next section, we introduce an alternative pseudo-GMM criterion function that depends on a set of incidental parameters distinct from $\boldsymbol{\delta}$. This new criterion enables the estimation of all structural parameters within a NP-GMM framework. Crucially, it retains the computational simplicity of the original pseudo-GMM approach.

\subsection{Our NP-GMM estimator \label{sec_our_NPGMM}}

\subsubsection{A new regression-like representation of the model}

Let $s_{0it}$ denote the market share of the outside good for consumer $i$ in market $t$. This consumer-level market share can be expressed as a function of the mean utilities $\boldsymbol{\delta}_{t}$ and $\boldsymbol{\sigma}$:
\begin{equation}
    s_{0it} \, = \, 
    \lambda_{it}\left(
        \boldsymbol{\delta}_{t}, \, \boldsymbol{\sigma}
    \right)
    \, = \,
    \displaystyle
    \frac{1}
    { 1 + \displaystyle \sum_{j=1}^J 
    \exp \Bigl( 
        \delta_{jt} + 
        \textstyle
        \sum_{k=1}^{K} x^{k}_{jt} \,
        [ \sigma^{k} \, \nu^{k}_{it} + 
        \boldsymbol{d}_{it}^{\prime} \, \boldsymbol{\pi}^{k} 
        ]
        \Bigr)
    }
\end{equation}
Combining the definition of $s_{0it}$ with equation \eqref{eq_model_mshares} for the model-implied market shares, we have:
\begin{equation}
    \begin{array}[c]{rcl}
    s_{jt} & = & 
    \displaystyle
    \frac{1}{N}
    \sum_{i=1}^{N}
    s_{0it} \; 
    \exp \Bigl( 
        \delta_{jt} + 
        \textstyle
        \sum_{k=1}^{K} x^{k}_{jt} \,
        [\sigma^{k} \, \nu^{k}_{it} + 
        \boldsymbol{d}_{it}^{\prime} \, \boldsymbol{\pi}^{k} ]
        \Bigr) \\ \\
    & = &
    \exp( \delta_{jt}) \; 
    \displaystyle
    \frac{1}{N}
    \sum_{i=1}^{N}
    \lambda_{it} \; 
    \exp \Bigl( 
        \textstyle
        \sum_{k=1}^{K} x^{k}_{jt} \,
        [\sigma^{k} \, \nu^{k}_{it} + 
        \boldsymbol{d}_{it}^{\prime} \, \boldsymbol{\pi}^{k} ]
        \Bigr) 
    \end{array}
\end{equation}
That implies the following regression-like expression:
\begin{equation}
    \ln \left( s_{jt} \right) 
    \; = \; 
    \boldsymbol{x}_{jt}^{\prime} \, \boldsymbol{\beta} 
    \, + \,
    h\left( 
        \boldsymbol{\lambda}_{t}, \boldsymbol{x}_{jt}, 
        \boldsymbol{\sigma}
    \right) 
    \, + \,
    \xi_{jt}
\label{eq_new_reg}
\end{equation}
where $\boldsymbol{\lambda}_{t}$ represents the vector of consumer-level market shares for the outside good, $(s_{0it}: \, i=1, 2, \dots N)$, and:\footnote{Note that when $\boldsymbol{\sigma} =0$, the function $h\left( \boldsymbol{\lambda}_{t}, \boldsymbol{x}_{jt}, \boldsymbol{\sigma} \right)$ simplifies to $\ln(s_{0t})$, and the regression equation \eqref{eq_new_reg} corresponds to the standard logit model.}
\begin{equation}
    h\left( 
        \boldsymbol{\lambda}_{t}, \boldsymbol{x}_{jt}, 
        \boldsymbol{\sigma}
    \right) 
    \; \equiv \; 
    \displaystyle
    \ln \left(
    \frac{1}{N}
    \sum_{i=1}^{N}
    \lambda_{it} \, 
    \exp \left\{ 
        \textstyle
        \sum_{k=1}^{K} x^{k}_{jt} \,
        [\sigma^{k} \, \nu^{k}_{it} + 
        \boldsymbol{d}_{it}^{\prime} \, \boldsymbol{\pi}^{k} ]
        \right\}
    \right).
\end{equation}

An interesting property of the vector $\boldsymbol{\lambda}_{t}$ is that—although it does not vary across products $j$—it captures all the information from the market shares of products other than $j$ that is required to compute the inverse demand function $\delta_{jt} = \symms^{-1}_{j}(\boldsymbol{s}_{t}, \boldsymbol{x}_{t}, \boldsymbol{\sigma})$. Specifically, the model implies:
\begin{equation} 
    \delta_{jt} \; = \; 
    \symms^{-1}_{j}(\boldsymbol{s}_{t}, \boldsymbol{x}_{t}, \boldsymbol{\sigma}) 
    \; = \;
    \ln \left( s_{jt} \right) -
    h\left( 
        \boldsymbol{\lambda}_{t}, \boldsymbol{x}_{jt}, 
        \boldsymbol{\sigma}
    \right) 
    \label{NPGMM}
\end{equation}
This property has an important implication for constructing the Pseudo-GMM criterion function used in our estimator: once we treat $\boldsymbol{\lambda}_{t}$ as a vector of incidental parameters, we no longer need to include $\boldsymbol{\delta}_{t}$ explicitly. Thus, conditional on $\boldsymbol{\lambda}_{t}$, recovering mean utilities $\boldsymbol{\delta}_{t}$ requires only evaluating the closed form expression $h\left( \boldsymbol{\lambda}_{t}, \boldsymbol{x}_{jt}, \boldsymbol{\sigma} \right)$ separately for each product, and does not require solving a $J$-dimensional fixed point.

Given the regression equation \eqref{eq_new_reg} and treating $\boldsymbol{\lambda} =(\boldsymbol{\lambda}_{t}: t=1, 2, \dots T)$ as incidental parameters, we can define the following pseudo-GMM criterion function:
{\small
\begin{equation}
    Q(\boldsymbol{\lambda}, \boldsymbol{\theta}) 
    \; = \;
    \displaystyle 
    \left(
    \sum_{j,t}
    \boldsymbol{z}_{jt} \, 
    \left[ 
        \ln \left( s_{jt} \right) - \boldsymbol{x}_{jt}^{\prime} \boldsymbol{\beta} -  h\left( 
        \boldsymbol{\lambda}_{t}, \boldsymbol{x}_{jt}, 
        \boldsymbol{\sigma}
        \right) 
    \right]
    \right)^{\prime} 
    \, \boldsymbol{W} \,
    \left(
    \sum_{j,t}
    \boldsymbol{z}_{jt} \, 
    \left[ 
        \ln \left( s_{jt} \right) - \boldsymbol{x}_{jt}^{\prime} \boldsymbol{\beta} -  h\left( 
        \boldsymbol{\lambda}_{t}, \boldsymbol{x}_{jt}, 
        \boldsymbol{\sigma}
        \right) 
    \right]
    \right)
\label{eq_our_criterion}
\end{equation}
} 

\begin{define} \label{def_our_npgmm}
    Our \textbf{NP-GMM estimator} is defined as a triple $(\widehat{\boldsymbol{\lambda}}, \, \widehat{\boldsymbol{\delta}}, \, \widehat{\boldsymbol{\theta}})$ satisfying the following three  conditions:
    \begin{itemize}
        \item NP-1. Given $\widehat{\boldsymbol{\lambda}}$, the parameter vector $\widehat{\boldsymbol{\theta}}$ minimizes the pseudo-GMM criterion function $Q$:
        \begin{equation}
            \widehat{\boldsymbol{\theta}} 
            \; = \;
            \arg \min_{\boldsymbol{\theta}} 
            \; \; 
            Q(\widehat{\boldsymbol{\lambda}}, \,
            \boldsymbol{\theta}) 
        \end{equation}

        \item NP-2. Given the structural parameters $\widehat{\boldsymbol{\theta}}$, the vector of mean utilities $\widehat{\boldsymbol{\delta}}_{t}$ solves, for each market $t$, the system of $J$ nonlinear equations implied by the demand system:
        \begin{equation}
            \boldsymbol{s}_{t}
            \; = \;
            \boldsymbol{\symms} 
            (\widehat{\boldsymbol{\delta}}_{t}, \, 
            \boldsymbol{x}_t, \,
            \widehat{\boldsymbol{\theta}}) 
            \qquad \text{for } 
            t = 1, \, 2, \, \dots, \, T.
        \end{equation}

        \item NP-3. Given $( \widehat{\boldsymbol{\delta}},  \widehat{\boldsymbol{\theta}})$, the vector $\widehat{\boldsymbol{\lambda}}$ contains the consumer-level market shares of the outside good which are implied by the model:
        \begin{equation}
            \widehat{\lambda}_{it} \, = \, 
            \lambda_{it} \left(
                \widehat{\boldsymbol{\delta}}_{t}, \, \widehat{\boldsymbol{\sigma}}
            \right)
            \qquad \text{for } 
            i = 1, \, 2, \, \dots, \, N,
            \text{ and } 
            t = 1, \, 2, \, \dots, \, T.
            \qquad \blacksquare
        \end{equation}    
    \end{itemize}
\end{define}

In the remainder of this section, we present two main results. First, we explain why the NP-GMM estimator differs from the standard GMM estimator. Second, we establish the consistency and asymptotic normality of the NP-GMM estimator under identification assumptions that are no stronger than those required for the consistency of the GMM estimator.

\subsubsection{The difference between GMM and NP-GMM estimators}

To illustrate the difference between the GMM and NP-GMM estimators, it is useful to define two mappings that relate the incidental parameters $\boldsymbol{\delta}$ and $\boldsymbol{\lambda}$ to the structural parameters $\boldsymbol{\theta}$. We begin by defining Berry's fixed-point mapping $\Psi$:
\begin{equation}
    \boldsymbol{\delta} 
    \; = \;
    \Psi\left( \boldsymbol{\delta}, \, \boldsymbol{\sigma} \right) 
    \; \equiv \;
    \boldsymbol{\delta} + 
    \ln\left( \boldsymbol{s} \right) - 
    \ln\left( 
        \boldsymbol{\symms} 
        \left( 
            \boldsymbol{\delta}, \boldsymbol{x}, 
            \boldsymbol{\sigma}
        \right)
    \right) 
\end{equation}
This mapping is continuously differentiable in both $\boldsymbol{\delta}$ and $\boldsymbol{\sigma}$. By \textit{Berry's Inversion Theorem} \citep{berry1994}, for any fixed value of $\boldsymbol{\sigma}$, the mapping $\Psi\left( \cdot , \boldsymbol{\sigma} \right)$ is a contraction. Consequently, there exists a unique solution $\boldsymbol{\delta}$ to the fixed-point problem for each $\boldsymbol{\sigma}$. We denote this solution by $\boldsymbol{\delta}^\ast(\boldsymbol{\sigma})$. Given $\boldsymbol{\sigma}$ and the associated vector of mean utilities $\boldsymbol{\delta}^\ast(\boldsymbol{\sigma})$, we can then define the mapping $\boldsymbol{\lambda}^\ast(\boldsymbol{\sigma})$, which yields the consumer-level market shares for the outside good:
\begin{equation}
    \boldsymbol{\lambda} 
    \; = \;
    \lambda^{\ast} \left( \boldsymbol{\sigma} \right) 
    \; \equiv \;
    \lambda \left( 
        \delta^{\ast}(\boldsymbol{\sigma}), \,
        \boldsymbol{\sigma}
    \right) 
\end{equation}

Taking into account the mappings $\delta^{\ast}$ and $\lambda^{\ast}$, we can establish a key relationship between the GMM and the NP-GMM criterion functions. For any vector $\boldsymbol{\sigma}$, we have:
\begin{equation}
    G \left( \boldsymbol{\theta} \right) 
    \; = \;
    G \left( \boldsymbol{\beta}, \,  
    \boldsymbol{\sigma} \right) 
    \; = \;
    Q \left( 
        \lambda^{\ast} \left( \boldsymbol{\sigma} \right), \,
        \boldsymbol{\theta}
    \right).
\end{equation}
This identity implies the following expression for the gradient of the GMM criterion function:
\begin{equation}
    \frac{\partial G}
    {\partial \boldsymbol{\beta}^{\prime}}      
    \; = \;
    \frac{\partial Q}
    {\partial \boldsymbol{\beta}^{\prime}}      
    \qquad \text{and} \qquad
    \frac{\partial G}
    {\partial \boldsymbol{\sigma}^{\prime}}      
    \; = \;
    \frac{\partial Q}{\partial \boldsymbol{\lambda}^{\prime}} 
    \cdot 
    \frac{\partial \lambda^{\ast} \left( \boldsymbol{\sigma} \right)}
    {\partial \boldsymbol{\sigma}^{\prime}} 
    \, + \,
    \frac{\partial Q}{\partial \boldsymbol{\sigma}^{\prime}}.
\label{eq_diff_gradients}
\end{equation}
It is important to note that neither 
$\partial Q/\partial \boldsymbol{\lambda}^{\prime}$ 
nor $\partial \lambda^{\ast}/\partial \boldsymbol{\sigma}^{\prime}$ 
are generally equal to zero. As a result, the condition 
$\partial G/\partial \boldsymbol{\sigma}^{\prime} = 0$  does not imply 
$\partial Q/\partial \boldsymbol{\sigma}^{\prime} = 0$, and vice versa. 

We are now in a position to formalize the distinction between the GMM and the NP-GMM estimators, denoted by $\widehat{\boldsymbol{\theta}}_{gmm}$ and $\widehat{\boldsymbol{\theta}}_{np}$, respectively. By construction, these estimators satisfy the following first-order optimality conditions:

\begin{equation}
    \left\{
    \begin{array}{rcl}
        \displaystyle
        \partial G \left( \widehat{\boldsymbol{\theta}}_{gmm} \right)
        / \partial \boldsymbol{\theta}  
        & = & 0 \\ 
        \displaystyle
        \partial Q \left( 
        \widehat{\boldsymbol{\lambda}}_{np}, \,
        \widehat{\boldsymbol{\theta}}_{np} \right)
        / \partial \boldsymbol{\theta} 
        & = & 0, \quad \text{with } \widehat{\boldsymbol{\lambda}}_{np} 
        = \lambda^{\ast}\left( \widehat{\boldsymbol{\theta}}_{np} \right)
    \end{array}
    \right.
\end{equation}
Using equation \eqref{eq_diff_gradients}, it follows that the NP-GMM estimator does not satisfy the first-order conditions defining the GMM estimator, and vice versa. Therefore, the two methods yield distinct estimators.

\subsubsection{Asymptotic properties of the NP-GMM estimator}

Throughout the paper, we use the hat notation $\widehat{\cdot}$ to denote statistics and functions that involve sampling error, and the subscript $_{0}$ to indicate their population counterparts. Let $\boldsymbol{\theta}_{0}$ be the true value the structural parameters in the population, and let $\boldsymbol{\lambda}_{0}$ be the corresponding value for $\boldsymbol{\lambda}$, i.e., $\boldsymbol{\lambda}_{0} = \lambda^{\ast}(\boldsymbol{\theta}_{0})$. Remember that the mapping $\lambda^{\ast}(\boldsymbol{\theta})$ is deterministic, i.e., it does not incorporate sampling/estimation error. Define the population counterpart of the sample criterion functions, for GMM and NP-GMM estimation: 
\begin{equation}
    G_{0} \left( 
        \boldsymbol{\theta} 
    \right) 
    \; \equiv \; 
    \mathbb{E} 
    \left[ 
        \, G \left( \boldsymbol{\theta} \right) \, 
    \right]
    \qquad \text{ and } \qquad
    Q_{0} \left( 
        \boldsymbol{\lambda}, \boldsymbol{\theta}
    \right) 
    \; \equiv \; 
    \mathbb{E} 
    \left[ 
        \, Q \left( \boldsymbol{\lambda}, \boldsymbol{\theta} \right) \, 
    \right]
\end{equation}

The following assumptions summarize the identification and regularity conditions needed for the consistency and asymptotic normality of the NP-GMM estimator. We study asymptotic properties of the estimators as either the number of products or the number of markets grows large, that is, as $JT \rightarrow \infty$.

\medskip 

\begin{assumption} \label{assumption_1} 
\textit{
    For any $\boldsymbol{\theta} \neq \boldsymbol{\theta}_{0}$, 
    $
        Q_{0} \left( 
            \lambda^{\ast}(\boldsymbol{\theta}), 
            \, \boldsymbol{\theta} 
        \right) 
        \; > \; 
        Q_{0} \left( 
            \boldsymbol{\lambda}_{0}, 
            \, \boldsymbol{\theta}_{0} 
        \right) = 0
    $. 
    Note that, by construction, this is equivalent to 
    $
        G_{0} \left( 
            \boldsymbol{\theta} 
        \right) 
        \; > \; 
        G_{0} \left( 
            \boldsymbol{\theta}_{0} 
        \right) = 0
    $ 
    for any $\boldsymbol{\theta} \neq \boldsymbol{\theta}_{0}$.
    $\qquad \blacksquare$
}
\end{assumption}

\medskip 

Assumption \ref{assumption_1} imposes standard conditions ensuring correct model specification and parameter identification. At the true parameter values, the moment conditions are satisfied, and the corresponding criterion function attains its minimum at zero. For any parameter value different from the true one, the moment conditions are violated, and the criterion function remains strictly positive. Importantly, this identification assumption holds equivalently under both standard GMM estimation, based on $G_{0}$, and the NP-GMM framework, based on $Q_{0}$.

Assumption \ref{assumption_2} presents standard regularity conditions which are used to derive the asymptotic distribution of the NP-GMM estimator.

\medskip 

\begin{assumption} \label{assumption_2} 
\textit{
    The following conditions hold: (a) $\Theta$ is a compact set. (b) $Q(\boldsymbol{\lambda}, \boldsymbol{\theta})$ is continuous and bounded. (c) $(\boldsymbol{\lambda}_0, \boldsymbol{\theta}_0) \in int([0,1]^{N} \times \Theta)$. (d) $(\boldsymbol{\lambda}_0, \boldsymbol{\theta}_0)$ is an isolated NP-GMM fixed point -- either it is unique, or there is a ball around it that does not contain any other NP-GMM fixed point. (e) $Q_{0}(\boldsymbol{\lambda}, \boldsymbol{\theta})$ is concave in $\boldsymbol{\theta}$ for any $\boldsymbol{\lambda}$ in a neighborhood around $\boldsymbol{\lambda}_{0}$.
    $\qquad \blacksquare$
}
\end{assumption}

\medskip 

\begin{prop} \label{proposition_1} 
    \textit{Under Assumptions \ref{assumption_1} and \ref{assumption_2}, as the sample size $JT$ goes to infinity, the NP-GMM estimator converges to the true value $(\boldsymbol{\lambda}_{0}, \boldsymbol{\theta}_{0})$ with probability approaching one. $\qquad \blacksquare$}
\end{prop}

\noindent \textit{Proof:} See Appendix \ref{proof_consistency}.

\medskip 

Proposition \ref{proposition_2} establishes the asymptotic distribution of the NP-GMM estimator. To present this result, it is useful to express the first-order conditions for the NP-GMM estimator as follows:
\begin{equation}
    \frac{\partial Q(\boldsymbol{\lambda}, \boldsymbol{\theta})}
    {\partial \boldsymbol{\theta}} 
    \; = \;
    \frac{1}{JT} \displaystyle
    \sum_{j,t} \, g_{jt}(\boldsymbol{\lambda}, \boldsymbol{\theta})
    \; = \;
    \frac{1}{JT} \displaystyle
    \sum_{j,t} \boldsymbol{z}_{jt}^{*}(\boldsymbol{\lambda}, \boldsymbol{\theta}) \, 
    \xi_{jt}(\boldsymbol{\lambda}, \boldsymbol{\theta})
    \; = \; 0
\end{equation}
Here, $\xi_{jt}(\boldsymbol{\lambda}, \boldsymbol{\theta})$ represents 
$\ln \left( s_{jt} \right) - \boldsymbol{x}_{jt}^{\prime} \boldsymbol{\beta} -  h\left( \boldsymbol{\lambda}_{t}, \boldsymbol{x}_{jt}, \boldsymbol{\sigma}
\right)$, and $\boldsymbol{z}_{jt}^{*}(\boldsymbol{\lambda}, \boldsymbol{\theta})$ denotes the effective vector of instruments—equal in dimension to the number of parameters in $\boldsymbol{\theta}$—obtained by multiplying the original instrument vector $\boldsymbol{z}_{jt}$ by the weighting matrix $\boldsymbol{W}$ and the gradient vector of $\xi_{jt}(\boldsymbol{\lambda}, \boldsymbol{\theta})$ with respect to $\boldsymbol{\theta}$. We also define the following matrices:
\begin{equation}
    \Omega_{\theta \theta} \equiv 
    \mathbb{E} \left[ 
        \frac{\partial g_{jt}(\boldsymbol{\lambda}_{0}, \boldsymbol{\theta}_{0})}
        {\partial \boldsymbol{\theta}} 
        \frac{\partial g_{jt}(\boldsymbol{\lambda}_{0}, \boldsymbol{\theta}_{0})}
        {\partial \boldsymbol{\theta}^{\prime}} 
    \right], \,
    \Omega_{\theta \lambda} \equiv 
    \mathbb{E} \left[ 
        \frac{\partial g_{jt}(\boldsymbol{\lambda}_{0}, \boldsymbol{\theta}_{0})}
        {\partial \boldsymbol{\theta}} 
        \frac{\partial g_{jt}(\boldsymbol{\lambda}_{0}, \boldsymbol{\theta}_{0})}
        {\partial \boldsymbol{\lambda}^{\prime}} 
    \right], \,
    \Lambda_{\theta} \equiv 
        \frac{\partial \boldsymbol{\lambda}^{\ast}(\boldsymbol{\theta}_{0})}
        {\partial \boldsymbol{\theta}^{\prime}} 
\end{equation}

\medskip 

\begin{prop} \label{proposition_2} 
    \textit{Under Assumptions \ref{assumption_1} and \ref{assumption_2}, the NP-GMM estimator has a normal asymptotic distribution: 
    $\sqrt{JT} \left( \widehat{\boldsymbol{\theta}} - \boldsymbol{\theta}_{0} \right)$ $\rightarrow_{d} N(0, \boldsymbol{V}_{np})$ with:
    \begin{equation}
        \boldsymbol{V}_{np} \, = \, 
        \left[ 
            \Omega_{\theta \theta} + \Omega_{\theta \lambda} \, \Lambda_{\theta}
        \right]^{-1} 
        \, \Omega_{\theta \theta} \,
        \left[ 
            \Omega_{\theta \theta} + \Lambda_{\theta}^{\prime} 
            \, \Omega_{\theta \lambda}^{\prime}
        \right]^{-1} 
        \qquad \blacksquare
    \label{eq_avar_np}
    \end{equation}
    }
\end{prop}

\noindent \textit{Proof:} See Appendix \ref{proof_consistency}.

\medskip

The asymptotic variance of the GMM estimator is given by $\Omega_{\theta \theta}^{-1}$. In this model, the Jacobian matrix $\Lambda_{\theta}$ is nonzero, and it is straightforward to show that this implies that matrix $\boldsymbol{V}_{np} - \Omega_{\theta \theta}^{-1}$ is positive definite. As a result, the NP-GMM estimator is asymptotically less efficient than the standard GMM estimator. Nevertheless, as shown in our Monte Carlo experiments, the resulting loss in efficiency is negligible in practice.

In the existing literature on Nested Pseudo estimation methods—including \citet{aguirregabiria_mira_2002, aguirregabiria_mira_2007}, as well as more recent applications such as \citet{lin2024endogeneity} and \citet{dearing2025efficient}—the asymptotic analysis is conducted under the maintained assumption that the number of incidental parameters is finite. This typically follows from a finite state space, which implies a fixed and finite number of conditional choice probabilities that does not grow with the sample size. In sharp contrast, the BLP model departs fundamentally from this setting: the dimension of the mean utility vector (the collection of $\delta$’s) increases with the number of products and therefore diverges as the sample size grows. As a result, existing proofs of consistency and asymptotic normality for NP-GMM or NPL estimators do not apply in this environment. This paper, specifically Propositions \ref{proposition_1} and \ref{proposition_2}, fills this gap by establishing consistency and asymptotic normality of the NP-GMM estimator in a setting where the number of incidental parameters grows without bound as the number of products tends to infinity, thereby providing a new asymptotic framework for NP-GMM estimation in high- and infinite-dimensional settings.

\subsection{Our NP-GMM algorithm \label{sec_npgmm_algorithm}}

This section presents the algorithm we use to compute our NP-GMM estimator. The algorithm is similar to the one described above for the Aguirregabiria-Mira method, but with two key differences. First, the pseudo-GMM criterion function $Q$ now depends on the full set of structural parameters. Second, before each pseudo-GMM iteration, we update the value of $\widehat{\boldsymbol{\lambda}}$ using its closed-form expression as a function of $(\widehat{\boldsymbol{\delta}}, \widehat{\boldsymbol{\sigma}})$.

\medskip

\noindent \textbf{NP-GMM algorithm:}

\begin{itemize}
    \item [0.] \textbf{Initialization:} The algorithm is initialized with values for the vector of mean utilities $\widehat{\boldsymbol{\delta}}^{0}$ and for the structural parameters associated to random coefficients,  $\widehat{\boldsymbol{\sigma}}^{0}$. For instance, we could start with the initial value  $\widehat{\boldsymbol{\sigma}}^{0} = 0$ and 
    $\widehat{\delta}^{0}_{jt} = \ln(s_{jt}) - \ln(s_{0t})$, that corresponds to the standard logit model.

    \item [1.] \textbf{Iteration $n \geq 1$:} It consists of a sequence of three steps.

    \begin{itemize}
        \item [Step 1.] \textbf{Updating $\widehat{\boldsymbol{\lambda}}$:} For any $i = 1, \, 2, \, \dots, \, N$ and $t = 1, \, 2, \, \dots, \, T$, we use $(\widehat{\boldsymbol{\delta}}^{n-1}, \widehat{\boldsymbol{\sigma}}^{n-1})$ to
        calculate $\widehat{\lambda}_{it}^{n}$ using its closed-form formula:
        \begin{equation}
            \widehat{\lambda}_{it}^{n}
            \, = \, 
            \lambda_{it}\left(
                \widehat{\boldsymbol{\delta}}_{t}^{n-1}, \, \widehat{\boldsymbol{\sigma}}^{n-1}
            \right)
            \, = \,
            \displaystyle
            \frac{1}
            { 1 + \sum_{j=1}^J 
            \exp \{ 
                \widehat{\delta}_{jt}^{n-1} + 
                \sum_{k=1}^{K} x^{k}_{jt} \,
                [ \widehat{\sigma}^{n-1}_{k} \, 
                \nu^{k}_{it} + 
                \boldsymbol{d}_{it}^{\prime} \, \widehat{\boldsymbol{\pi}}^{n-1}_{k} 
                ]
            \}
            }
        \end{equation}
        
        \item [Step 2.] \textbf{Pseudo-GMM estimation of $\boldsymbol{\theta}$:} Given $\widehat{\boldsymbol{\lambda}}^{n}$, we obtain $\widehat{\boldsymbol{\theta}}^{n}$ as:
        \begin{equation}
            \widehat{\boldsymbol{\theta}}^{n}
            \; = \;
            \arg \min_{\boldsymbol{\theta}} 
            \; \; 
            Q \left(
                \widehat{\boldsymbol{\lambda}}^{n}, \boldsymbol{\theta} 
            \right) .
    \end{equation}

    \item [Step 3.] \textbf{Newton-Kantorovich iteration} to update mean utilities:
    \begin{equation}
    \widehat{\boldsymbol{\delta}}^{n}_{t}
    \; = \;
    \widehat{\boldsymbol{\delta}}^{n-1}_{t} + 
    \left[ 
        \nabla_{\delta^{\prime}} \ln\boldsymbol{\symms}(\widehat{\boldsymbol{\delta}}^{n-1}_{t}, \widehat{\boldsymbol{\sigma}}^{n})
    \right]^{-1} \, 
    \left[ 
        \ln \boldsymbol{s}_{t} - 
        \ln \boldsymbol{\symms}(\widehat{\boldsymbol{\delta}}^{n-1}_{t}, \widehat{\boldsymbol{\sigma}}^{n})
    \right] 
    \; \text{for } 
        t = 1, \, 2, \, \dots, \, T. \label{step_3_npgmm}
    \end{equation}
\end{itemize}

\end{itemize}

Upon convergence, the algorithm provides the NP-GMM estimator defined by conditions NP-1 to NP-3.

\medskip

\noindent \textbf{Remark 1:} The main computational advantage of this algorithm over NFXP is that it eliminates the need to repeatedly solve the $T$ fixed-point problems. Instead, the NP-GMM algorithm solves these problems only once, at the point of convergence. This efficiency comes at the cost of solving multiple pseudo-GMM minimization problems. However, these minimizations are relatively inexpensive, as the pseudo-GMM criterion has a simple closed-form expression and is globally concave in $\boldsymbol{\theta}$.\footnote{As with any existing method for estimating the BLP model, our algorithm yields a local optimum upon convergence, which may not be the global optimum if multiple local optima exist. The standard practice applies: the algorithm should be run from multiple starting values to assess the presence of multiple local minima. In such cases, the estimator is defined as the local optimum that achieves the lowest value of the GMM criterion function.} 

In contrast to the \textit{Approximate BLP} method proposed by \citet{ABLP2015}—which we describe in Section \ref{sec_ablp}—our criterion function does not rely on a linear approximation of the nonlinear demand system. Instead, it leverages the structure of the model and treats certain endogenous objects as incidental parameters. As our Monte Carlo experiments demonstrate, this difference has important implications for the relative performance of the two methods.

\medskip

\noindent \textbf{Remark 2:} Step 1, which updates $\widehat{\boldsymbol{\lambda}}$, is relatively inexpensive and can be fully parallelized across $NT$ threads. The pseudo-GMM iteration in Step 2 is slightly more computationally demanding than in the standard Aguirregabiria-Mira version of the NP-GMM. This is because the criterion function is no longer quadratic, and we no longer have a closed-form expression for the estimator. However, the criterion function still has a simple analytical form and can be concentrated in $\boldsymbol{\sigma}$ by exploiting its quadratic structure in $\boldsymbol{\beta}$. Most importantly, the function is easy to evaluate, globally convex in $\boldsymbol{\theta}$, and has a closed-form expression for its gradient. As a result, standard gradient-based optimization methods can be used to reliably obtain the global minimum, i.e., the pseudo-GMM estimator.

The Newton–Kantorovich outer iteration in Step 3 can be replaced by a standard fixed-point iteration based on Berry’s mapping. We implemented this alternative version of the algorithm as well. In our experiments, the version using Newton iterations performs markedly better: it requires far fewer outer iterations to converge and reduces the incidence of non-convergence to essentially zero. 

\subsection{The Approximate-BLP estimator \label{sec_ablp}}

\citet{ABLP2015} proposed an estimation procedure for the BLP model parameters based on a first-order approximation of the log-demand system around an arbitrary vector of mean utilities, denoted by $\boldsymbol{\delta}_{0}$. For any value $(\boldsymbol{\delta}, \boldsymbol{\sigma})$:
\begin{equation}
    \ln \boldsymbol{\symms}(\boldsymbol{\delta}, \boldsymbol{\sigma}) 
    \, \approx \, 
    \ln \boldsymbol{\symms}(\boldsymbol{\delta}_{0}, \boldsymbol{\sigma}) +
    \nabla_{\delta^{\prime}} \ln\boldsymbol{\symms}(\boldsymbol{\delta}_{0}, \boldsymbol{\sigma}) \, 
    \left( \boldsymbol{\delta} - \boldsymbol{\delta}_{0} \right)
\end{equation}
where $\approx$ denotes a first-order Taylor series expansion, $\nabla_{\delta'} \equiv \partial/\partial \delta'$, and $   \nabla_{\delta^{\prime}} \ln\boldsymbol{\symms}(\boldsymbol{\delta}_{0}, \boldsymbol{\sigma})$ is a $J \times J$ Jacobian matrix. 

Thanks to this linear approximation, the inversion of the system of demand equations to obtain the vector of average utilities $\boldsymbol{\delta}$ can be approximated using matrix inversion or linear projection. We denote this approximation to the inversion problem using the operator $\Psi^{ablp}(\boldsymbol{\delta}_{0}, \boldsymbol{\sigma})$. That is:
\begin{equation}
    \boldsymbol{\delta} \; = \; 
    \Psi^{ablp}(\boldsymbol{\delta}_{0}, \boldsymbol{\sigma})
    \; \equiv \;
    \boldsymbol{\delta}_{0} + 
    \left[ 
        \nabla_{\delta^{\prime}} \ln\boldsymbol{\symms}(\boldsymbol{\delta}_{0}, \boldsymbol{\sigma})
    \right]^{-1} \, 
    \left[ 
        \ln \boldsymbol{s} - 
        \ln \boldsymbol{\symms}(\boldsymbol{\delta}_{0}, \boldsymbol{\sigma})
    \right]
\label{ABLP}
\end{equation}

The \textit{Approximate-BLP} (ABLP) estimator is based on the following criterion function:
\begin{equation}
    Q^{ablp}(\boldsymbol{\delta}_{0}, \boldsymbol{\theta}) 
    \; = \;
    \displaystyle 
    \left(
    \sum_{j,t}
    \boldsymbol{z}_{jt} \, 
    \left[ 
        \Psi^{ablp}_{jt}(\boldsymbol{\delta}_{0}, \boldsymbol{\sigma}) -
        \boldsymbol{x}_{jt}^{\prime} \, \boldsymbol{\beta} 
    \right]
    \right)^{\prime} 
    \, \boldsymbol{W} \,
    \left(
    \sum_{j,t}
    \boldsymbol{z}_{jt} \, 
    \left[ 
        \Psi^{ablp}_{jt}(\boldsymbol{\delta}_{0}, \boldsymbol{\sigma}) -
        \boldsymbol{x}_{jt}^{\prime} \, \boldsymbol{\beta} 
    \right]
    \right)
\label{eq_ablp_criterion}
\end{equation}
where $\Psi^{ablp}_{jt}(\boldsymbol{\delta}_{0}, \boldsymbol{\sigma})$ represents element $jt$ of the vector $\Psi^{ablp}(\boldsymbol{\delta}_{0}, \boldsymbol{\sigma})$. 

\begin{define} \label{def_ablp_estimator}
    The \textbf{ABLP estimator} is defined as a pair $(\widehat{\boldsymbol{\delta}}, \widehat{\boldsymbol{\theta}})$ satisfying two sets of conditions:
    \begin{itemize}
        \item ABLP-1. Given $\widehat{\boldsymbol{\delta}}$, we have that $\widehat{\boldsymbol{\theta}} = \arg \min_{\boldsymbol{\theta}} Q^{ablp}(\widehat{\boldsymbol{\delta}}, \boldsymbol{\theta})$. 

        \item ABLP-2. Given $\widehat{\boldsymbol{\theta}}$, we have that $\widehat{\boldsymbol{\delta}} = \Psi^{ablp}(\widehat{\boldsymbol{\delta}}, \widehat{\boldsymbol{\sigma}})$.
        $\qquad \blacksquare$
    \end{itemize}
\end{define}

\medskip

\noindent \textbf{Approximate BLP algorithm:} \citet{ABLP2015} proposed the following algorithm to compute this ABLP estimator.

\begin{itemize}
    \item [0.] \textbf{Initialization:} The algorithm is initialized with values for the vector of mean utilities $\widehat{\boldsymbol{\delta}}^{0}$.

    \item [1.] \textbf{Iteration $n \geq 1$:} It consists of a sequence of two steps.

    \begin{itemize}
        \item [Step 1.] \textbf{Pseudo-GMM estimation of $\boldsymbol{\theta}$:} Given $\widehat{\boldsymbol{\delta}}^{n-1}$, we obtain $\widehat{\boldsymbol{\theta}}^{n}$ as:
        \begin{equation}
            \widehat{\boldsymbol{\theta}}^{n}
            \; = \;
            \arg \min_{\boldsymbol{\theta}} 
            \; \; 
            Q^{ablp} \left(
                \widehat{\boldsymbol{\delta}}^{n}, \boldsymbol{\theta} 
            \right) .
        \end{equation}
    
    \item [Step 2.] \textbf{Updating $\widehat{\boldsymbol{\delta}}$.} Mean utilities are updated using the approximate mapping $\Psi^{ablp}$:
    \begin{equation}
        \widehat{\boldsymbol{\delta}}^{n} = 
        \Psi^{ablp}(\widehat{\boldsymbol{\delta}}^{n-1}, \widehat{\boldsymbol{\sigma}}^{n}).
    \end{equation}
\end{itemize}

\end{itemize}

Upon convergence, the algorithm provides the ABLP estimator defined by conditions ABLP-1 and ABLP-2.

\begin{table}[ht]
\centering
\caption{Features of Algorithms for the Estimation of BLP Demand Model}
\label{tab_feaures_algorthms}
\begin{tabular}{r|cccc}
\hline \hline \\ 
\multicolumn{1}{r|}{\textbf{Features}} 
& \textbf{NFXP} & \textbf{MPEC} 
& \textbf{ABLP} & \textbf{NP-GMM} \\ 
\hline \\ 
\multicolumn{1}{r|}{\textit{Avoids solving inversion at each $\sigma$ trial value}} 
& No & Yes & Approx. & Yes \\ 
& & & & \\ 
\multicolumn{1}{r|}{\textit{Avoids  linearization of shares mapping}} 
& Yes & Yes & No & Yes \\ 
& & & & \\ 
\multicolumn{1}{r|}{\textit{Closed-form inversion conditional on auxiliary objects}} 
& No & No & No & Yes \\ 
& & & & \\ 
\multicolumn{1}{r|}{\textit{Closed-form gradient of criterion function}} 
& No & No & No & Yes \\ 
& & & & \\ 
\multicolumn{1}{r|}{\textit{Scales well with multithreading}} 
& Mixed & Mixed & Weak & Strong \\ 
& & & & \\ 
\multicolumn{1}{r|}{\textit{Returns GMM estimator}} 
& Yes & Yes & No & No \\ 
& & & & \\ 
\multicolumn{1}{r|}{\textit{Asymptotically equivalent to GMM estimator}} 
& Yes & Yes & Yes & No \\ 
& & & & \\ 
\hline \hline
\end{tabular}
\end{table}

\medskip 

\noindent \textbf{Remark 3:} There is a fundamental difference between the Pseudo-GMM estimation in the ABLP algorithm and that in our NP approach. Each Pseudo-GMM step in the ABLP procedure requires repeatedly evaluating $\Psi^{ablp}_{jt}(\boldsymbol{\delta}_{0}, \boldsymbol{\sigma})$ —once for each trial value of $\boldsymbol{\sigma}$ —which, in turn, involves computing the matrix inverse $\left[ \nabla_{\delta^{\prime}} \ln \boldsymbol{\symms}(\boldsymbol{\delta}_{0}, \boldsymbol{\sigma}) \right]^{-1}$ for each of those values. By contrast, the Pseudo-GMM step in our NP-GMM algorithm only requires evaluating simple, closed-form expressions for each product individually when computing both the criterion function and its gradient.

\medskip

\noindent \textbf{Remark 4:} The distinction between ABLP and our NP-GMM  algorithm becomes even more significant when parallel computing is introduced. A key advantage of our approach is the straightforward computation of the sample counterparts of the error terms $\xi_{jt}$, which can be efficiently parallelized across multiple threads. As demonstrated in our numerical experiments in Section \ref{sec_comparing_gmm}, this leads to substantial reductions in CPU time during the Pseudo-GMM estimation step.\footnote{Appendix \ref{anaytic gradient} derives the gradients of our Pseudo-GMM objective function, which have a notably simple form and are straightforward to vectorize.} In contrast, parallel computing yields negligible time savings for the Pseudo-GMM step in the ABLP method. This is because evaluating ABLP's criterion function and its gradients requires inverting matrices that depend on all products simultaneously, limiting the effectiveness of parallelization across products.

Table \ref{tab_feaures_algorthms} summarizes key features of the main algorithms used for the estimation of BLP demand models, and highlights the novel features of the NP-GMM method.

\section{Monte Carlo Experiments} \label{sec:monte}

In this section, we use simulated datasets to compare the computational speed of our approach with that of the ABLP method. \citet{ABLP2015} show that the ABLP algorithm is significantly faster than both the NFXP and MPEC algorithms, especially in large-sample settings.

\subsection{Data-Generating Process}

We adopt the same Monte Carlo simulation setting as in \citet{dube2012improving}, which is also used by \citet{ABLP2015} to compare the ABLP and MPEC approaches.\footnote{All experiments are conducted on a system equipped with an Intel(R) Core(TM) i7-8850H CPU @ 2.60GHz, 16GB of RAM, running Windows 10 (64-bit) and MATLAB R2023a. For additional speed gains, we also provide Julia code implementing our method.}

We conduct a series of Monte Carlo experiments using samples in which  the number of markets $T$ takes values $50$, $100$, or $500$, and the number of products $J$ takes values $25$, $50$, $100$, $200$, $400$, $800$, $1600$, and $3200$. In each market $t$, all $J$ products are present.

Let the deterministic component of utility be defined as: $\boldsymbol{x}_{jt}^{\prime} \boldsymbol{\beta} = \beta_0 + \beta_1 x_{1,j} + \beta_2 x_{2,j} + \beta_3 x_{3,j} + \beta_p p_{jt}$, where the vector of coefficients is given by $\boldsymbol{\beta}^{\prime} = (\beta_0, \beta_1, \beta_2, \beta_3, \beta_p) = (0, 1.5, 1.5, 0.5, -3)$. The three exogenous product characteristics -—$x_{1,j}$, $x_{2,j}$, and $x_{3,j}$-— are constant across markets $t$ and are drawn from a multivariate normal distribution given by:
\begin{equation}
    \begin{bmatrix}
        \boldsymbol{x}_{1j} \\ \boldsymbol{x}_{2j} \\ \boldsymbol{x}_{3j} 
    \end{bmatrix}
    \sim N \left(
        \begin{bmatrix}
            0 \\ 0 \\ 0  
        \end{bmatrix},
        \begin{bmatrix}
            1 & -0.8 & 0.3 \\
            -0.8 & 1 & 0.3 \\
            0.3 & 0.3 & 1  
        \end{bmatrix} 
    \right).
\end{equation}

In each market $t$, we simulate $N = 1,000$ individual-level taste draws, denoted by $\boldsymbol{\nu}_{it}^{\prime}$ $= (\nu_{0,it}, \nu_{1,it}, \nu_{2,it}, \nu_{3,it}, \nu_{p,it})$, which are independently drawn from a standard normal distribution. The scale parameters associated with these consumer-level random coefficients are $\boldsymbol{\sigma}^{\prime} =$ $(\sigma_0, \sigma_1, \sigma_2, \sigma_3, \sigma_p)$ $= (\sqrt{0.5}, \sqrt{0.5}, \sqrt{0.5}, \sqrt{0.5}, \sqrt{0.2})$. Note that this DGP abstracts from observed demographics, $\boldsymbol{d}_{i}$.

Each product $j$ has a market-specific vertical characteristic $\xi_{jt}$, which is independently drawn across products and markets from a standard normal distribution. Product prices vary across markets and are determined by the following equation:
\begin{equation}
    p_{jt} = 3 + 
    x_{1,j} + x_{2,j} + x_{3,j} +
    1.5 \, \xi_{jt} + 
    5 \, \omega_{jt},
\end{equation}
where $\omega_{jt}$ is a marginal cost shock, independently drawn from a uniform distribution $U(0,1)$ across products and markets. Note that price $p_{jt}$ is an endogenous regressor in the demand estimation due to its correlation with the unobserved demand shock $\xi_{jt}$.

We address price endogeneity using instrumental variables derived from observable marginal cost shifters. Specifically, the DGP includes six observable cost shifters per product, denoted by ${w_{k,jt}}$, which are generated according to the following equation:
\begin{equation}
    w_{k,jt} \, = \, 
    0.25 \cdot \left|
        5 \, \omega_{jt} + 
        1.1 (x_{1,j} + x_{2,j} + x_{3,j}) 
    \right| + 
    e_{k,jt}, 
    \quad \text{for } k = 1, 2, \dots, 6,
\end{equation}
where $e_{k,jt} \sim i.i.d. \, U(0,1)$. By construction, each instrument $w_{k,jt}$ is strongly correlated with price $p_{jt}$ and independent of the unobserved component $\xi$.

The full vector of instrumental variables $\boldsymbol{z}_{jt}$ includes 42 variables, comprising: a constant term (1 variable); the product's own $x$ variables in levels, squares, and cubes (9 variables); the product's own $w$ variables in levels, squares, and cubes (18 variables); the product of all own $x$ variables (1); the product of all own $w$ variables (1); the cross-products between $x_{1,j}$ and each $w$ variable (6); and the cross-products between $x_{2,j}$ and each $w$ variable (6).

\subsection{Comparing computational speeds of NP-GMM and ABLP methods \label{sec_comparing_gmm}}

We begin by comparing the two methods in terms of both the computational speed of the algorithms and the statistical properties of the estimators. For a fair comparison, we apply the same convergence criterion and identical initial values to both algorithms. Convergence is assessed using the infinity norm of the differences between successive iterations of the vectors $(\widehat{\boldsymbol{\delta}}^{n} - \widehat{\boldsymbol{\delta}}^{n-1})$ and $(\widehat{\boldsymbol{\theta}}^{n} - \widehat{\boldsymbol{\theta}}^{n-1})$. Specifically, an algorithm is considered to have converged when
\begin{equation}
    \max_{j,t} |\widehat{\delta}_{jt}^{n} - \widehat{\delta}_{jt}^{n-1}| < 10^{-6} \quad \text{and} \quad \max_{k} |\widehat{\theta}_{k}^{n} - \widehat{\theta}_{k}^{n-1}| < 10^{-6}.
\end{equation}

Since both NP-GMM and ABLP may yield multiple fixed points, we implement each algorithm using five distinct starting values—identical across the two methods—for each simulated dataset. For each algorithm, we select the fixed point that minimizes the value of the criterion function. To generate the initial values $\widehat{\boldsymbol{\delta}}^{0}$ and $\widehat{\boldsymbol{\sigma}}^{0}$, we follow the approach proposed by \citeauthor{dube2012improving} (\citeyear{dube2012improving}) to construct five random starting points.\footnote{To generate a random initial value for $\widehat{\boldsymbol{\sigma}}^{0}$, \citeauthor{dube2012improving} (\citeyear{dube2012improving}) use the formula: $\widehat{\sigma}_{k}^{0} = 0.5 \cdot |\widehat{\beta}_{k}^{logit}| \cdot \widetilde{U}_{k}$, where  $\widehat{\boldsymbol{\beta}}^{logit}$ is the IV estimate of $\boldsymbol{\beta}$ 
from a standard logit model with instruments $\boldsymbol{z}_{jt}$. $\widetilde{U}_{k}$ is a random draw from a U(0,1). 
Given $\widehat{\boldsymbol{\sigma}}^{0}$, the initial value $\widehat{\boldsymbol{\delta}}^{0}$ is obtained by solving the fixed-point problem that defines the vector of mean utilities.} For each sample and algorithm, we define a successful estimation as one in which at least one of the five starting points converges.

\subsubsection{Total wall-clock time and its components}

Table \ref{tab:perform_compare} and Figure \ref{fig:speed_compare} compare the performance of the ABLP and NP-GMM methods using simulated samples with $T = 100$ markets and the numbers of products ranging from $J=25$ to $J=400$. Each method is implemented with its optimal number of computational threads.\footnote{See Section 
\ref{sec_optimal_number_threads} on the calculation of the optimal number of threads.} Across all values of $J$, NP-GMM achieves significantly lower wall-clock times than ABLP (Panel A).\footnote{\textit{Wall-clock time} is the standard computer-science measure of the total time required to complete a task, as recorded by a real clock. It includes all delays, such as input/output operations, system interruptions, and waiting time. This contrasts with \textit{CPU time}, which measures only the periods during which the processor is actively executing instructions.} This advantage exists even though NP-GMM requires more outer iterations (Panel D), more inner iterations (Panel E), and more criterion-function evaluations (Panel F). The source of this advantage is clearly shown in  Panel B of Table \ref{tab:perform_compare}. ABLP incurs a substantially higher computational cost in each inner-loop GMM minimization because, at every trial value of $\sigma$, it repeatedly computes and inverts a Jacobian matrix. By conditioning on the incidental parameters $\boldsymbol{\lambda}$, the NP-GMM algorithm avoids this expensive step, leading to a far lower per-iteration cost and ultimately much faster overall computation despite requiring more iterations. 

Figure \ref{fig:speed_compare} shows how the performance gap between ABLP and NP-GMM methods widens more than proportionally as $J$ increases. An OLS regression of log–Time-per-Inner-Iteration on log-$J$ yields a slope of $0.96$ for the NP-GMM algorithm —indicating roughly linear growth— and a slope of $1.32$ for the ABLP algorithm. Consequently, the relative speed advantage of NP-GMM becomes increasingly pronounced as the size of the problem grows. For example, extrapolating these rates to $J = 12{,}800$ products yields predicted times per-inner-iteration of approximately $1076$ seconds for ABLP versus only $33$ seconds for NP-GMM.

\begin{table}[H]\centering
\caption{Comparing Computational Performance ($T = 100$)}
\label{tab:perform_compare}
%\footnotesize

\begin{tabular}{c cc cc cc}
\hline \hline

% =================== Row of panels: A E F ===================
& \multicolumn{2}{c}{\textbf{Panel A}} 
& \multicolumn{2}{c}{\textbf{Panel B}} 
& \multicolumn{2}{c}{\textbf{Panel C}} \\
& \multicolumn{2}{c}{Wall-clock Time}
& \multicolumn{2}{c}{Time per Inner Iteration}
& \multicolumn{2}{c}{Time per Crit. Fun. Eval.} \\
\cmidrule(lr){2-3}\cmidrule(lr){4-5}\cmidrule(lr){6-7}
$J$ & ABLP & NP-GMM & ABLP & NP-GMM & ABLP & NP-GMM \\
\hline
25  & 55.84  & 16.95  & 0.3619  & 0.0817  & 0.2764 & 0.0633 \\
50  & 74.61  & 40.52  & 0.6543  & 0.1596  & 0.4905 & 0.1242 \\
100 & 161.83 & 92.18  & 1.4040  & 0.3008  & 1.0579 & 0.2379 \\
200 & 612.55 & 249.75 & 4.9949  & 0.6127  & 3.8164 & 0.4731 \\
400 & 2,415.2 & 740.07 & 12.4464 & 1.1790  & 9.5191 & 0.9068 \\
\hline

% =================== Row of panels: B C D ===================
& \multicolumn{2}{c}{\textbf{Panel D}} 
& \multicolumn{2}{c}{\textbf{Panel E}} 
& \multicolumn{2}{c}{\textbf{Panel F}} \\
& \multicolumn{2}{c}{Number of Outer Iterations}
& \multicolumn{2}{c}{Number of Inner Iterations}
& \multicolumn{2}{c}{Number of Crit. Fun. Eval.} \\
\cmidrule(lr){2-3}\cmidrule(lr){4-5}\cmidrule(lr){6-7}
$J$ & ABLP & NP-GMM & ABLP & NP-GMM & ABLP & NP-GMM \\
\hline
25  & 9    & 11.4 & 151.5 & 202.2 & 199.4 & 260.2 \\
50  & 7.6  & 14.2 & 113.5 & 247.0 & 150.8 & 316.6 \\
100 & 7.3  & 16.1 & 114.7 & 297.4 & 151.8 & 375.1 \\
200 & 7.6  & 21.6 & 122.6 & 392.7 & 159.8 & 509.1 \\
400 & 10.6 & 31.0 & 191.8 & 595.0 & 250.8 & 772.4 \\
\hline \hline
\end{tabular}

\vspace{0.5em}

\begin{minipage}{0.9\textwidth}
\footnotesize
\textit{Note:} For each value of $J$, we generate five datasets with $T = 100$ and $N = 1,000$. The estimation for each dataset is done with five different starting points. Reported means are based on the $5 \times 5 = 25$ runs. All time measures are wall-clock seconds.
\end{minipage}

\end{table}

\begin{figure}[H]
\centering
    \caption{Total Wall-Clock Times of ABLP and BP-GMM Methods}
    \includegraphics[scale=0.60]{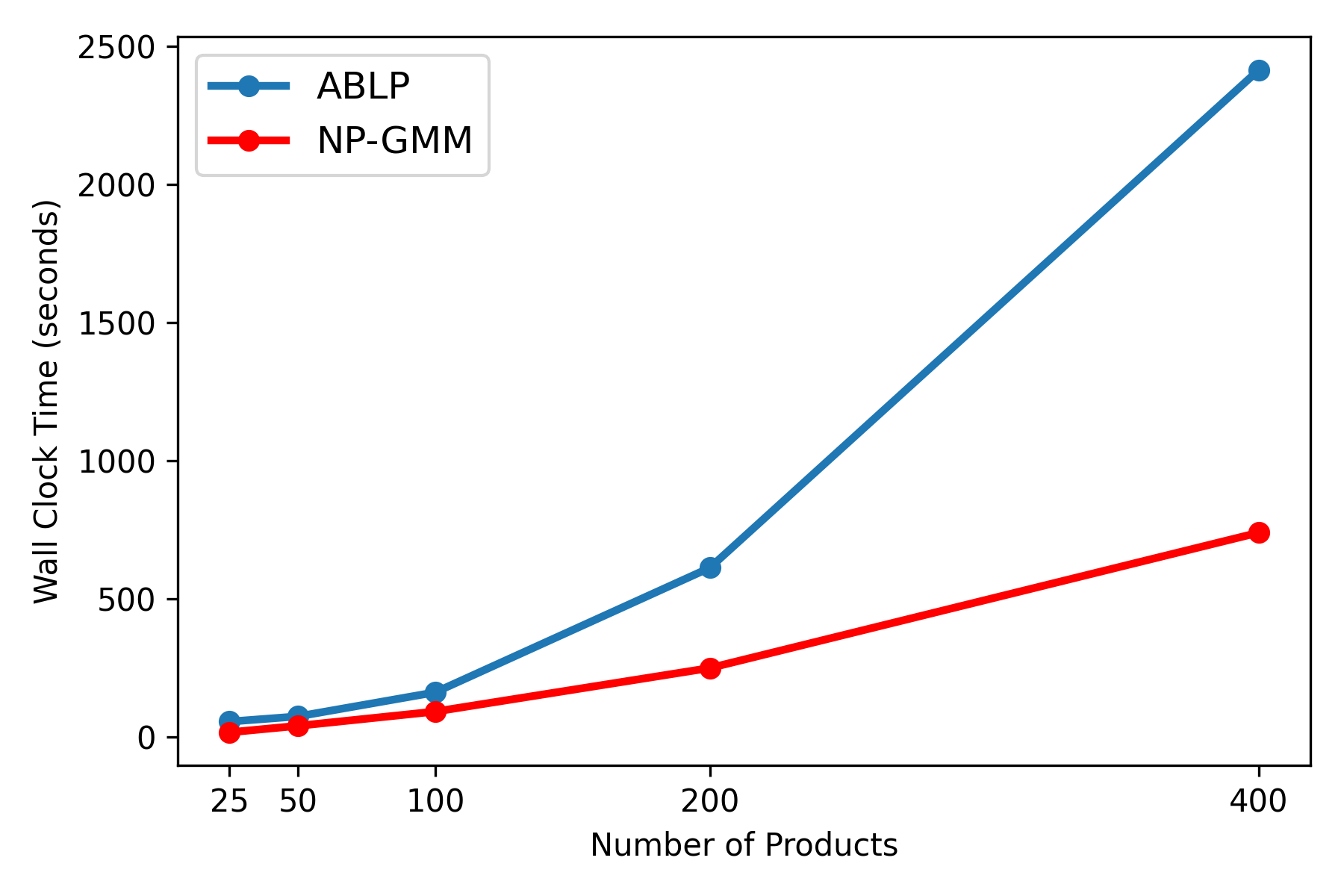}
    \fnote{\textit{Note:} For each value of $J$, we generate five datasets with $T = 100$ and $N = 1,000$. The estimation for each dataset is done with five different starting points. Reported means are based on the $5 \times 5 = 25$ runs.}
\label{fig:speed_compare}
\end{figure}

\subsubsection{Parallel computing and optimal number of threads \label{sec_optimal_number_threads}}

Figure \ref{fig:eva_time_threads} reports the average wall-clock time required by the NP-GMM and ABLP algorithms to evaluate their respective GMM criterion functions and gradients. The averages are computed over 1,000 simulated samples with $J=25$, $T=500$, and $N=1{,}000$. All evaluations are conducted using parallel computing, with the number of threads varying from 1 to 6. For the evaluation of the GMM objective function, NP-GMM and ABLP differ only in their computation of $\boldsymbol{\delta}$. NP-GMM relies on Equation \eqref{NPGMM}, which provides a simple closed-form mapping in $\boldsymbol{\sigma}$ for a given $\boldsymbol{\lambda}$. In contrast, ABLP uses Equation \eqref{ABLP}, a mapping in $\boldsymbol{\sigma}$ that crucially requires matrix inversion for its calculation. Because $\boldsymbol{\delta}$ is a simpler closed-form function of $\boldsymbol{\sigma}$ under NP-GMM, our estimator is substantially faster than ABLP in computing both the GMM objective function and its associated gradients.

A further advantage of our market-share inversion is its ability to exploit parallel computing more effectively. The \citet{ABLP2015} MATLAB code uses the \texttt{-singleCompThread} flag or option to disable multithreading, in which case CPU time and wall-clock time are comparable across algorithms. When multiple threads are allowed, however, we observe substantial reductions in wall-clock time for evaluating the GMM objective function and its gradients under NP-GMM. By contrast, ABLP’s gradient evaluation does not benefit from multithreading and in fact deteriorates as the number of threads increases, as shown in Figure \ref{fig:eva_time_threads}.

\begin{figure}[H]
    \centering
    \captionsetup{justification=centering}
    \caption{Speed Comparison of NP-GMM and ABLP Algorithms -- \\ Computing GMM Criteria and Gradient}
    \includegraphics[scale=0.16]{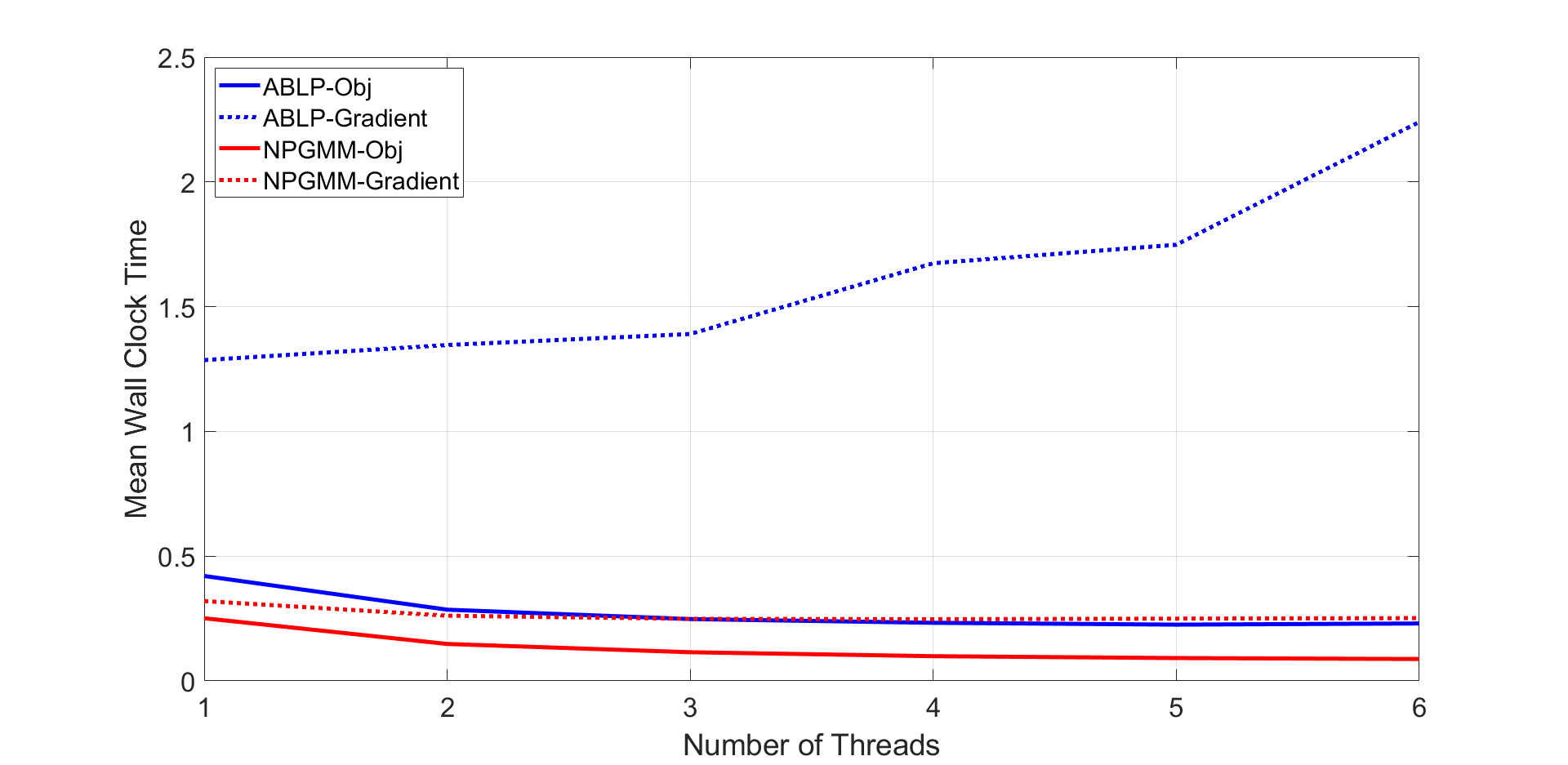}
    \fnote{\textit{Note:} Average wall-clock time (in seconds) for 1,000 evaluations of the GMM objective function and associated gradients. The model setting is $J$ = 25, $T$ = 500, and $N$ = 1,000.}
\label{fig:eva_time_threads}
\end{figure}

To understand why the optimal number of threads in parallel computing is finite—and why this optimum is substantially larger for NP-GMM than for ABLP—it is useful to consider the following simple decomposition. Let $\ell$ denote the number of threads, $T(\ell)$ the total execution time using $\ell$ threads, $T_{\text{parallel}}$ the time required to perform tasks that can be parallelized, $T_{\text{serial}}$ the time required for inherently serial tasks, and $T_{\text{overhead}}(\ell)$ the overhead introduced by parallelization, which is increasing in $\ell$. Total execution time can then be written as
\begin{equation}
    T(\ell) \; = \; 
    \frac{T_{\text{parallel}}}{\ell}
    \; + \;
    T_{\text{serial}}
    \; + \;
    T_{\text{overhead}}(\ell)
\end{equation}

As the number of threads increases, parallel execution incurs overhead costs $T_{\text{overhead}}(\ell)$ that are absent in serial computation, including thread creation and management, synchronization, inter-thread communication, and operating-system scheduling. These overhead costs grow with the number of threads and eventually offset the gains from dividing the parallel workload across more processors. For example, if the overhead is linear in $\ell$, so that $T_{\text{overhead}}(\ell)=k_{0}+k_{1} \, \ell$, the optimal number of threads is finite and given by $\ell^{*} =\sqrt{T_{\text{parallel}}/k_{1}}$.

This framework also clarifies why the optimal number of threads is substantially larger for NP-GMM than for ABLP. Under NP-GMM, the computation of the $\delta$’s and the associated gradients relies on the closed-form expression in Equation \eqref{NPGMM}. These calculations are fully parallelizable and therefore contribute almost entirely to $T_{\text{parallel}}$. In contrast, under ABLP the computation of the $\delta$’s and gradients is based on Equation \eqref{ABLP}, which requires matrix inversion. Such operations are only partially parallelizable and involve substantial serial components, implying that a larger fraction of the computation falls into $T_{\text{serial}}$. As a result, NP-GMM can efficiently exploit a larger number of threads before overhead dominates, leading to a higher optimal degree of parallelization than in ABLP.

Given this, we determine the optimal number of threads for each estimator to minimize overall evaluation time and apply it to ensure a fair comparison. Under the model setting $J = 25$, $T = 500$, and $N = 1,000$, the optimal number of threads is 2 for ABLP and 6 for NP-GMM. See Appendix \ref{omega faster} for further details on optimal thread configurations under different model settings. Using these optimized settings, we employ MATLAB's profiler to track the wall-clock time for 1,000 evaluations of the GMM objective function and associated gradients, as shown in Figure \ref{fig:eva_time}. NP-GMM demonstrates significant speed advantages in computing both the GMM objective function and its gradients. Overall, our NP-GMM algorithm is approximately five times faster than ABLP in evaluating the objective function and gradients. See Appendix \ref{omega faster} for more details.

\begin{figure}[H]
\centering
    \includegraphics[scale=0.25]{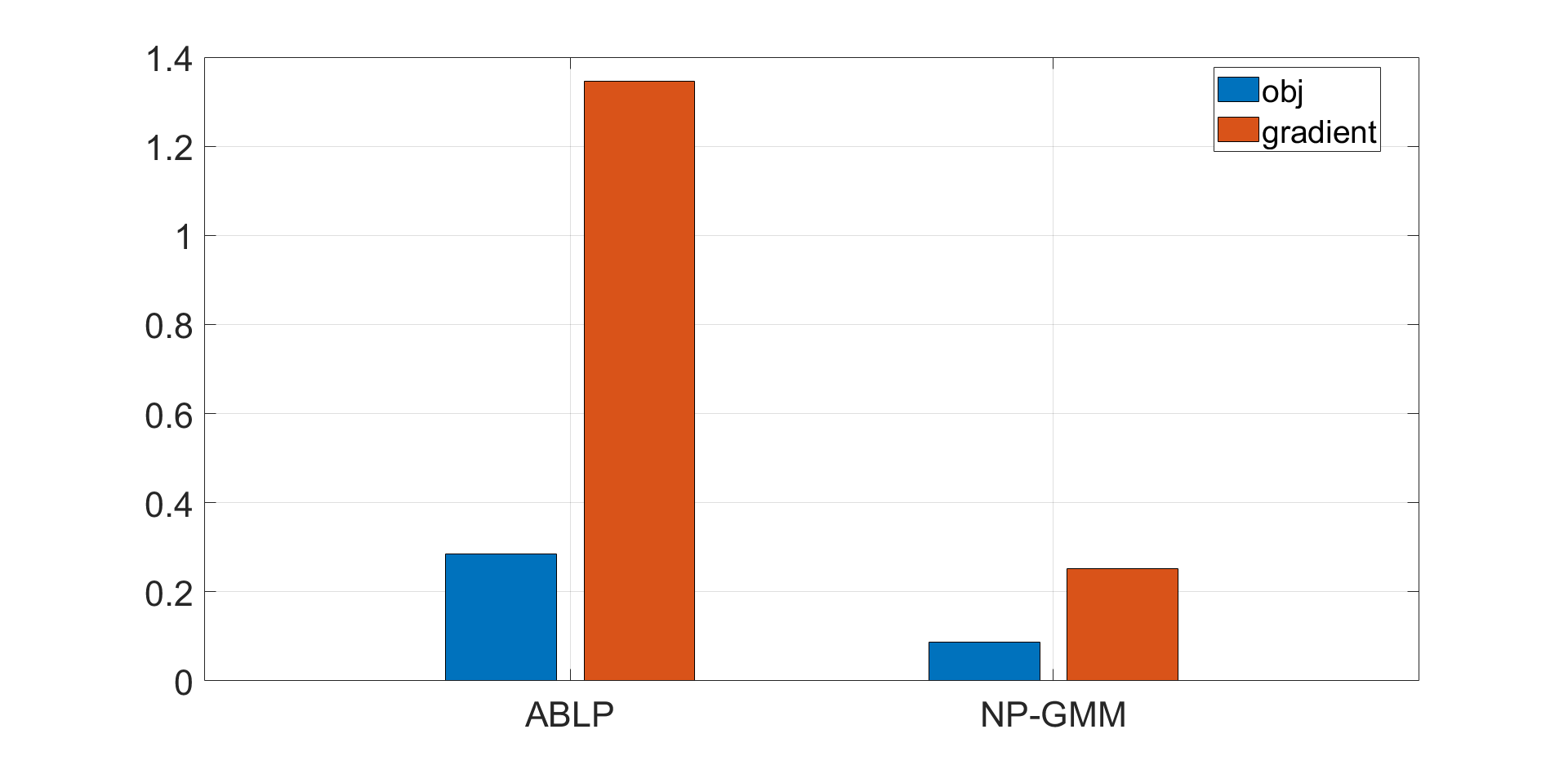}
    \caption{Speed Comparison}
    \fnote{\textit{Note:} Average wall-clock time (in seconds) for 1,000 evaluations of the GMM objective function and associated gradients. The model setting is $J$ = 25, $T$ = 500, and $N$ = 1,000. We use the optimal number of threads: 2 for ABLP and 6 for NP-GMM.}
\label{fig:eva_time}
\end{figure} 

\subsubsection{Comparing Statistical Properties}

Tables \ref{tab_mc_experiments} and \ref{tab:combined_results} summarize the performance of the two estimators and algorithms in two Monte Carlo experiments: one with $J = 25$ and $T = 500$, and another with $J = 100$ and $T = 100$. Additional experiments with comparable sample sizes $J \times T$ but different combinations of $J$ and $T$ yield similar statistical properties for both estimators.\footnote{In our DGP, and thus in all our experiments, the number of random draws for consumer-level random coefficients is fixed at $N = 1,000$.} Each experiment is based on 20 simulated datasets. For each dataset, estimation is initialized from five random starting values, resulting in a total of 100 estimations per method. The maximum number of outer iterations is capped at 100, and the algorithm terminates once this limit is reached.

Table \ref{tab_mc_experiments} summarizes the computational and statistical performance of the two algorithms in these experiment. Both algorithms exhibit excellent convergence behavior. For each experiment and dataset, at least one run converged for each estimator. In the experiment with $J=25$ and $T=500$, out of 100 estimations per algorithm (20 datasets × 5 starting values), ABLP failed to converge in 2 cases (from a single dataset), while NP-GMM failed in 4 cases (from two datasets). In the experiment with $J=100$ and $T=100$, each algorithm failed to converge in one case.

\begin{table}[h]
\centering
\caption{Monte Carlo Experiments: Performance Comparisons} \label{tab_mc_experiments}
\begin{tabular}{l|cc|cc}
    \hline \hline
    & \multicolumn{2}{c|}{\textbf{With $J = 25$ \& $T = 500$}}
    & \multicolumn{2}{c}{\textbf{With $J = 100$ \& $T = 100$}} \\ 
    & \textbf{ABLP} & \textbf{NP-GMM} 
    & \textbf{ABLP} & \textbf{NP-GMM} \\ \hline
    \textbf{Computational Properties} 
    & & & & \\ 
    & & & & \\ 
    \textit{Convergence Rate (\%)} 
    & & & &
    \\
    \multicolumn{1}{r|}{\textit{At dataset level:}}  
    & 100.0 & 100.0 & 100.0 & 100.0 \\
    \multicolumn{1}{r|}{\textit{At dataset-initial value level:}}  
    & 98.0 & 96.0 & 99.0 & 99.0 \\ 
    & & & & \\ 
    \textit{Wall-clock Time (in seconds):} 
    &  &  &  & \\
    \multicolumn{1}{r|}{\textit{Mean}} 
    & 263  & 101  &  230 &  93 \\
    \multicolumn{1}{r|}{\textit{25th pct}}
    & 207  & 78  &  159 &  66 \\
    \multicolumn{1}{r|}{\textit{Median}}
    & 248 & 99 & 181 & 86 \\
    \multicolumn{1}{r|}{\textit{75th pct}}
    & 311 & 117  & 215 & 101 \\ 
    & & & & \\ 
    \hline 
    \multicolumn{1}{l|}{\textbf{Statistical Properties:}}  
    & & & & \\ 
    & & & & \\ 
    \multicolumn{1}{r|}{\textit{Root Mean Square Error:}}    
    & 0.1101  & 0.1086 & 0.2282 &  0.1744 \\
    \multicolumn{1}{r|}{\textit{Mean Absolute Bias of Price Coeff:}}    
    & 0.0423 & 0.0368 & 0.0181 & 0.0396 \\
    \multicolumn{1}{r|}{\textit{Standard deviation of Price Coeff.:}}    
    & 0.1049 & 0.1016 & 0.1490 & 0.1582 \\
    & & & & \\ 
    \hline \hline
    \multicolumn{5}{l}{\footnotesize \textit{Note:} Convergence Rate is the percentage of estimations meeting convergence criterion.}
\end{tabular}
\end{table}

Consistent with the results reported above, the wall-clock time required to implement NP-GMM is between two and three times lower than that of ABLP. This difference appears systematically across all datasets and initial parameter values.

The bottom panel of Table \ref{tab_mc_experiments} reports the statistical properties of the estimators. The results show that both estimators exhibit very similar finite-sample performance in terms of bias, variance, and mean squared error. The NP-GMM estimator performs slightly better in terms of RMSE in the two experiments.

Table \ref{tab:combined_results} reports detailed Monte Carlo results on the means and standard deviations of the parameter estimates under the two methods. For each parameter, the two methods exhibit very similar finite-sample biases and variances.

\begin{table}[H]\centering
\caption{Summary Statistics of Parameter Estimates}
\label{tab:combined_results}
\footnotesize
\begin{tabular}{cccccccccccc}
\hline\hline

 & & $\beta_0$ & $\beta_1$ & $\beta_2$ & $\beta_3$ & $\beta_p$ & $\sigma_0$ & $\sigma_1$ & $\sigma_2$ & $\sigma_3$ & $\sigma_p$ \\
\cmidrule(lr){3-7}\cmidrule(lr){8-12}
 & True & 0 & 1.5 & 1.5 & 0.5 & -3 & 0.7071 & 0.7071 & 0.7071 & 0.7071 & 0.4472 \\
\hline
%==================== Panel A ====================%
\multicolumn{12}{l}{\textbf{Panel A}} \\
\multicolumn{12}{l}{\textbf{$J=25$, $T=500$}} \\
\cmidrule(lr){3-7}\cmidrule(lr){8-12}
\multirow{2}{*}{ABLP} 
& Mean & -0.0940 & 1.4809 & 1.4801 & 0.4967 & -2.9577 & 0.6661 & 0.7047 & 0.7072 & 0.7153 & 0.4371 \\
& Std  &  0.1579 & 0.0539 & 0.0532 & 0.0376 & 0.1049 & 0.2507 & 0.0349 & 0.0240 & 0.0229 & 0.0349 \\
\cmidrule(lr){3-7}\cmidrule(lr){8-12}
\multirow{2}{*}{NP-GMM} 
& Mean & -0.0814 & 1.4813 & 1.4797 & 0.4990 & -2.9632 & 0.6423 & 0.7094 & 0.7069 & 0.7181 & 0.4393 \\
& Std  &  0.1626 & 0.0525 & 0.0529 & 0.0374 & 0.1016 & 0.2473 & 0.0332 & 0.0274 & 0.0215 & 0.0333 \\
\hline

%==================== Panel B ====================%
\multicolumn{12}{l}{\textbf{Panel B}} \\
\multicolumn{12}{l}{\textbf{$J=100$, $T=100$}} \\
\cmidrule(lr){3-7}\cmidrule(lr){8-12}
\multirow{2}{*}{ABLP} 
& Mean & -0.0841 & 1.4916 & 1.4940 & 0.4988 & -2.9819 & 0.8272 & 0.7050 & 0.7018 & 0.7228 & 0.4422 \\
& Std  &  0.2211 & 0.0573 & 0.0602 & 0.0464 & 0.1490 & 0.4612 & 0.0367 & 0.0402 & 0.0335 & 0.0474 \\
\cmidrule(lr){3-7}\cmidrule(lr){8-12}
\multirow{2}{*}{NP-GMM} 
& Mean & -0.0849 & 1.4904 & 1.4910 & 0.4995 & -2.9604 & 0.6306 & 0.7086 & 0.6975 & 0.7231 & 0.4344 \\
& Std  &  0.2395 & 0.0552 & 0.0564 & 0.0453 & 0.1582 & 0.4179 & 0.0386 & 0.0436 & 0.0324 & 0.0504 \\
\hline\hline

\end{tabular}
\end{table}

\section{Application: Wine Demand at the LCBO} 
\label{sec:lcbo}

In this section, we use wine demand data from the LCBO to illustrate the implementation of the proposed methodology and to assess its computational performance relative to ABLP. The empirical results are intended to demonstrate the scalability and feasibility of estimating demand models with a large number of products. The LCBO is a provincially owned near monopoly that oversees the wholesale and retail distribution of alcoholic beverages in Ontario. A detailed description of the industry and the raw data is provided in \citeauthor{aguirregabiria2023decentralized} (\citeyear{aguirregabiria2023decentralized}).

\subsection{Data and Working Sample}

The dataset covers all 634 LCBO stores in Ontario and contains monthly sales data from October 2011 to October 2013, for a total of 25 months. LCBO classifies stores into six size categories --AAA, AA, A, B, C, and D-- from largest to smallest. We focus on the five AAA stores, which are the highest-volume outlets and together account for a substantial share of total LCBO wine sales. Restricting attention to these stores allows us to study demand in large, stable markets with a broad and consistently available product assortment.

We abstract from consumers’ store choice and treat each store-month observation as a separate market. Over the sample period, the LCBO sold 14,364 distinct wine products. Many of these products are offered only temporarily or in a limited number of stores, which creates substantial product turnover and sparsity at the store-month level. To avoid unbalanced product availability and to simplify the interpretation and estimation of the demand system, we restrict attention to the 76 store–month markets (out of 125) in which at least 3,000 products record positive sales. Within each selected market, we designate the products with the highest sales as inside goods and aggregate sales of all remaining wine products into the outside good. We fix the number of inside goods at $J=2900$, thereby trimming products with extremely small market shares. Our full working sample therefore consists of the top 2,900 wine products in each of the 76 markets, yielding a total of 220,400 product–market observations.

Product-level variables include standardized monthly sales (measured in 750-ml bottle equivalents), standardized price per 750-ml bottle, alcohol content, sugar content, and indicators for red and white wine, with other wine types serving as the omitted category. We define market size as twice the maximum monthly sales observed at each store over the sample period. Table \ref{tab:summary_lcbo} presents summary statistics for our working sample and the variables in our model.\footnote{The LCBO also sells non-wine spirits which we exclude from our sample. We rely on the LCBO’s own classification of products as “wine”. Ten out of the 14,364 products that LCBO classifies as wines have alcohol content above 30\%. One of these—Zwack Unicum Slivovitz 3-Year-Old—is sufficiently popular to appear in our estimation sample, which explains why the maximum observed alcohol content is 47\%. All other wines in the sample have alcohol levels below 23\%.}

\begin{table}[htbp]
\centering
\caption{Summary Statistics: Market Shares and Product Characteristics}
\label{tab:summary_lcbo}
\begin{tabular}{lccccc}
\hline\hline
Variable & Observations & Mean & Std.\ Dev. & Min & Max \\
\hline
Market share & 220{,}400 & $1.08\times10^{-4}$ & $2.31\times10^{-4}$ & $1.98\times10^{-6}$ & 0.0113 \\
Price (CAD \$ per 750ml bottle)
& 220{,}400 & 21.53 & 19.57 & 4.93 & 199.85 \\
Alcohol (\%) 
& 220{,}400 & 13.07 & 1.71  & 4.80 & 47.00 \\
Sugar (Grams per litter)
& 220{,}400 & 12.52 & 25.36 & 2.00 & 373.00 \\
Red wine (Dummy)
& 220{,}400 & 0.55  & 0.50  & 0    & 1 \\
White wine (Dummy)
& 220{,}400 & 0.31  & 0.46  & 0    & 1 \\
\hline\hline
\end{tabular}

\vspace{0.5em}
\begin{minipage}{1.00\textwidth}
\footnotesize
\textit{Note:} The sample includes $J=2{,}900$ products across $T=76$ markets. Market shares are computed using market size defined as twice the maximum monthly sales observed at each store. 
\end{minipage}
\end{table}

\subsection{Empirical Model}

In our empirical model, the product attribute vector $\boldsymbol{x}_{jt}$ includes a constant, alcohol content, sugar content, price, and indicators for red and white wine. The deterministic component of utility is:
\begin{equation}
    \boldsymbol{x}_{jt}^{\prime} \boldsymbol{\beta}
    \; = \; 
    \beta_0 + \beta_{\text{alc}} \; \text{alc}_j + \beta_{\text{sugar}} \; \text{sugar}_j 
    + \beta_{\text{red}} \; \mathbf{1}\{j \text{ is red}\}
    + \beta_{\text{white}} \; \mathbf{1}\{j \text{ is white}\}
    + \beta_p \; p_{jt}.
\end{equation}
We allow for random coefficients on the constant, alcohol content, sugar content, and price. In each market $t$, individual taste heterogeneity is captured by draws $\boldsymbol{\nu}_{it}
= (\nu_{0,it}, \nu_{\text{alc},it}, \nu_{\text{sugar},it}, \nu_{p,it})$, which are independently drawn from a standard normal distribution. We simulate 1,000 consumers in each market. The corresponding scale parameters are denoted by $\boldsymbol{\sigma}
= (\sigma_0, \sigma_{\text{alc}}, \sigma_{\text{sugar}}, \sigma_p)$. As in our Monte Carlo experiments, we abstract from observed demographics.

We treat price as endogenous. In addition to the product characteristics summarized in Table~\ref{tab:summary_lcbo}, the data include information on each wine’s country of origin and a broad classification of wine style (e.g., dry or sweet). We construct BLP-style instruments using the average alcohol and sugar content of other wines from the same country of origin, as well as those of other wines within the same style category. 

These instruments are motivated by cost-side and technological similarities among wines that share a country of origin or style. Alcohol and sugar content are closely linked to grape characteristics and fermentation technology, reflecting exogenous production conditions that systematically affect wine quality and pricing outcomes \citep{combris1997estimation,ashenfelter2008predicting}. Accordingly, variation in the average alcohol and sugar content of competing wines within narrowly defined competitive groups acts as a cost shifter that predicts equilibrium prices, consistent with standard differentiated-product pricing models in which firms’ optimal prices respond to rivals’ cost conditions \citep{berry_levinshon_1995,nevo2001measuring}. Empirically, these instruments are highly predictive of price, with a first-stage $F$-statistic of 1{,}903.98, suggesting that weak-instrument concerns are unlikely.

\subsection{Estimation Results}

To study how the relative performance of the algorithms varies with the number of products, we estimate the model using different values of $J$, holding the set of markets $T$ fixed at the 76 identified above. We report results for the full sample with $J = 2{,}900$ products, as well as for smaller subsamples with $J = 800$ and $J = 1{,}600$ products, drawn at random from the full product set.

Both ABLP and NP-GMM benefit from parallelization as $J$ increases. Table~\ref{tab:per_eval_time_lcbo} in Appendix~\ref{estimate_details_lcbo} reports mean wall-clock time per evaluation of the objective function and its gradients using one to six threads. For NP-GMM, six threads consistently yield the fastest evaluation. For ABLP, five threads are optimal when $J = 800$, and six threads are optimal when $J = 1{,}600$ or $2{,}900$. These thread settings are used in all subsequent estimations.

Estimation is conducted with five random starting points, with both estimators initialized from the same starting values and updating the mean utility vector $\boldsymbol{\delta}$ via Newton's method in the outer loops. Table \ref{tab:lcbo_estimates} reports parameter estimates obtained using the two methods for each of the three samples. Across all samples and parameters, except for $\beta_0$ and $\sigma_0$, the estimates produced by the two methods are extremely close, with only negligible differences between them. Parameters $\beta_0$ and $\sigma_0$ are typically weakly identified in this setting due to limited variation in baseline preferences once observed characteristics and unobserved product quality are accounted for (see, e.g., \citet{dube2012improving} and \citet{ABLP2015}).

The near-identical point estimates produced by ABLP-GMM and NP-GMM across all configurations underscore that the results are not driven by algorithmic choices. Overall, the model delivers a plausible and robust characterization of wine demand at the LCBO. Across all product sets, the estimates paint a very coherent picture of consumer preferences in the LCBO wine market that is both economically intuitive and remarkably stable as the number of products grows. 

\begin{table}[H]\centering
\caption{Parameter Estimates: LCBO Application}
\label{tab:lcbo_estimates}
\resizebox{0.98\textwidth}{!}{
\begin{tabular}{cccccccccccc}
\hline\hline
$J$ & & $\beta_0$ & $\beta_{alc}$ & $\beta_{sugar}$ &
$\beta_{p}$ & $\beta_{red}$ & $\beta_{white}$ &
$\sigma_0$ & $\sigma_{alc}$ & $\sigma_{sugar}$ & $\sigma_{p}$ \\
\hline
\multirow{2}{*}{800}
& ABLP   & -10.0973 & -0.1116 & -0.0591 & -0.2219 & 0.2370 & 0.1666 & 0.4533 & 0.2407 & 0.0180 & 0.0117 \\
& NP-GMM & -10.0954 & -0.1115 & -0.0592 & -0.2219 & 0.2370 & 0.1666 & 0.4476 & 0.2407 & 0.0180 & 0.0120 \\
\hline
\multirow{2}{*}{1,600}
& ABLP   & -10.7204 & -0.1251 & -0.0548 & -0.2063 & 0.2484 & 0.1860 & 1.8423 & 0.2156 & 0.0025 & 0.0276 \\
& NP-GMM & -13.0064 & -0.0937 & -0.0520 & -0.2197 & 0.2286 & 0.1625 & 4.6914 & 0.2068 & 0.0081 & 0.0077 \\
\hline
\multirow{2}{*}{2,900}
& ABLP   & -9.8265 & -0.1114 & -0.0466 & -0.2253 & 0.2487 & 0.1779 & 0.8842 & 0.1967 & 0.0086 & 0.0237 \\
& NP-GMM & -11.1340 & -0.0940 & -0.0474 & -0.2271 & 0.2326 & 0.1576 & 3.7913 & 0.2039 & 0.0025 & 0.0006 \\
\hline\hline
\end{tabular}
}
\end{table}

Both alcohol and sugar content enter average utility with negative coefficients, suggesting that consumers on average dislike higher-alcohol and sweeter wines. Price has an economically meaningful negative effect on utility, generating substantial price sensitivity. The magnitude is stable across product counts. The positive coefficients on the red and white wine dummies indicate that, relative to the omitted “other wine” category, both red and white wines enjoy a significant utility premium. The random-coefficient estimates reveal important heterogeneity in tastes. There is meaningful heterogeneity in preferences for alcohol content and price, indicating that while the average consumer dislikes high alcohol and high prices, there is a non-trivial subset of consumers who are less price sensitive and who positively value stronger wines. By contrast, the estimated dispersion in sugar preferences is small, suggesting relatively homogeneous attitudes toward sweetness. 

Figure~\ref{fig:speed_compare_lcbo} shows total wall-clock times, illustrating that the performance gap between ABLP and NP-GMM widens disproportionately as $J$ increases.

Table~\ref{tab:lcbo_perform} reports computational performance metrics, highlighting the speed advantage of NP-GMM over ABLP. Although NP-GMM requires more objective function evaluations, its per-evaluation cost is substantially lower. As $J$ doubles, NP-GMM's per-evaluation time roughly doubles, whereas ABLP's increases by a factor of three. Across the three values of $J$, evaluations of the ABLP objective function and gradients are between 12 and 41 times slower than those of NP-GMM. Overall, NP-GMM is between 7 and 14 times faster than ABLP in this application.

\begin{figure}[H]
\centering
\caption{Total Wall-Clock Times: LCBO Application}
\includegraphics[scale=0.22]{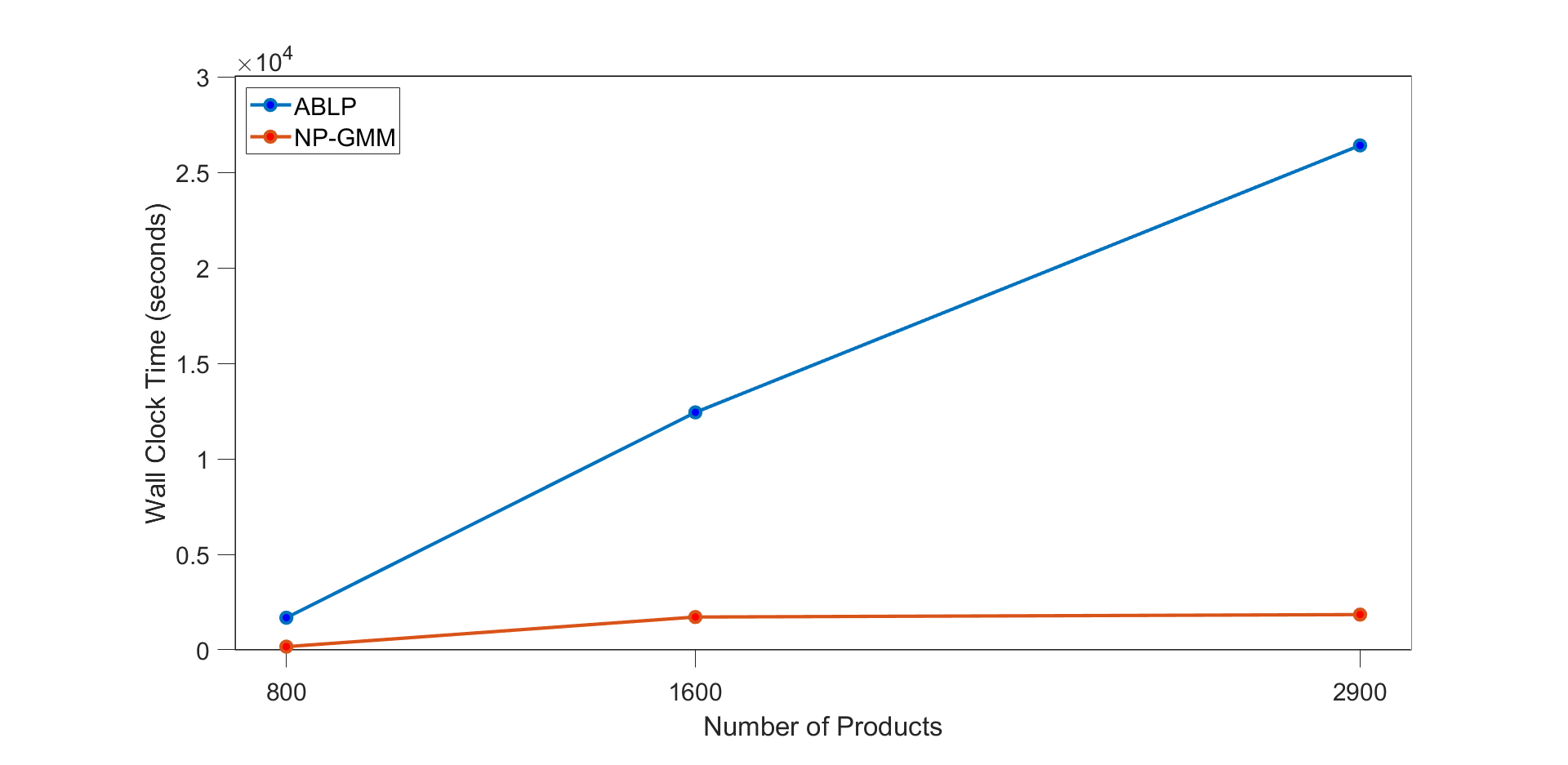}
\fnote{\textit{Note:} Each estimation is conducted with five random starting points. Reported means are based on five runs. Each $J$ is paired with $T = 76$ and $N = 1{,}000$.}
\label{fig:speed_compare_lcbo}
\end{figure} 

\begin{table}[H]\centering
\caption{Comparing Computational Performance: LCBO Application}
\label{tab:lcbo_perform}
\begin{tabular}{c cc cc cc}
\hline \hline
& \multicolumn{2}{c}{\textbf{Panel A}} 
& \multicolumn{2}{c}{\textbf{Panel B}} 
& \multicolumn{2}{c}{\textbf{Panel C}} \\
& \multicolumn{2}{c}{Wall-clock Time}
& \multicolumn{2}{c}{Time per Inner Iteration}
& \multicolumn{2}{c}{Time per Crit. Fun. Eval.} \\
\cmidrule(lr){2-3}\cmidrule(lr){4-5}\cmidrule(lr){6-7}
$J$ & ABLP & NP-GMM & ABLP & NP-GMM & ABLP & NP-GMM \\
\hline
800   & 1{,}683.02  & 168.43   & 20.99  & 1.64  & 14.62 & 1.17 \\
1,600  & 12{,}431.66 & 1{,}719.89 & 71.19  & 3.22  & 53.37 & 2.38 \\
2,900  & 26{,}405.95 & 1{,}847.03 & 262.09 & 6.34  & 183.96 & 4.48 \\
\hline
& \multicolumn{2}{c}{\textbf{Panel D}} 
& \multicolumn{2}{c}{\textbf{Panel E}} 
& \multicolumn{2}{c}{\textbf{Panel F}} \\
& \multicolumn{2}{c}{Number of Outer Iterations}
& \multicolumn{2}{c}{Number of Inner Iterations}
& \multicolumn{2}{c}{Number of Crit. Fun. Eval.} \\
\cmidrule(lr){2-3}\cmidrule(lr){4-5}\cmidrule(lr){6-7}
$J$ & ABLP & NP-GMM & ABLP & NP-GMM & ABLP & NP-GMM \\
\hline
800   & 5.0  & 6.2  & 79.4  & 91.4  & 114.0 & 127.4 \\
1,600  & 8.8  & 22.2 & 177.8 & 473.6 & 232.2 & 641.0 \\
2,900  & 6.2  & 13.2 & 99.8  & 201.8 & 142.0 & 285.4 \\
\hline \hline
\end{tabular}

\vspace{0.5em}
\begin{minipage}{0.9\textwidth}
\footnotesize
\textit{Note:} Each estimation is conducted with 5 random starting points. Reported values are means over 5 runs. All time measures are in seconds. Each model uses $T = 76$ markets and $N = 1{,}000$ simulated consumers.
\end{minipage}
\end{table}

\section{Conclusion}
\label{sec:conclusion}

This paper proposes a new computational approach to estimating the random-coefficients logit demand model that substantially reduces the computational burden of existing methods while preserving their desirable large-sample properties. The key insight is that, conditional on consumer-level probabilities of choosing the outside option, the BLP share inversion admits a closed-form representation that decouples the evaluation of unobserved product characteristics across products. Exploiting this representation allows us to reformulate estimation as a nested pseudo-GMM problem in which the most computationally demanding component of standard BLP estimation—the repeated numerical solution of the share inversion—is no longer required at each parameter trial.

This reformulation leads to an estimation algorithm with three important features. First, it replaces repeated fixed-point iterations or constrained optimization over mean utilities with simple, regression-like calculations that are analytically tractable and admit closed-form gradients. Second, it shifts computational effort from the inner loop to the outer loop of the algorithm, yielding a criterion function that is considerably cheaper to evaluate even if a larger number of outer iterations is required. Third, because the resulting objective and gradient evaluations are naturally separable across products and simulation draws, the algorithm scales particularly well with parallel computing.

Monte Carlo experiments and an empirical application show that these features translate into substantial computational gains at no cost in terms of finite-sample performance. Even in moderately sized demand systems, the proposed estimator is approximately twice as fast as the fastest currently available alternative, and its relative advantage grows more than proportionally with the number of products. Importantly, these gains do not come at the expense of statistical reliability: across all designs considered, the pseudo-GMM estimator exhibits sampling variability and finite-sample bias that are very close to those of fully iterated GMM estimators.

Beyond its immediate computational benefits, the proposed approach opens several avenues for future research. First, the idea of conditioning on auxiliary statistics that summarize competitive interactions may be useful in other models of differentiated-product demand, including extensions with richer heterogeneity, alternative choice sets, or dynamic consumer behavior. Second, the algorithm’s favorable scaling properties make it well suited for applications involving large product spaces, high-frequency data, or repeated estimation in simulation-based policy analysis. 

\clearpage

\baselineskip14pt

\appendix
\renewcommand{\theequation}{\thesection\arabic{equation}}

\begin{appendices}

\setstretch{1.5}

\section{Proofs of Propositions \ref{proposition_1} and \ref{proposition_2} \label{proof_consistency}}

\subsection{Definitions}

Throughout this appendix, we use the hat notation $\widehat{\cdot}$ to denote statistics and functions that involve sampling error, and the subscript $_{0}$ to indicate their population counterparts.

Define the set of sample NP-GMM fixed points as:
\[
    \widehat{\Upsilon} \equiv 
    \left\{ 
    (\boldsymbol{\lambda}, \boldsymbol{\theta}) 
    \in [0, 1]^N \times \Theta : \, 
    \boldsymbol{\theta} = \widehat{\theta}^{\ast}(\boldsymbol{\lambda}) 
    \text{ and } 
    \boldsymbol{\lambda} = \lambda^{\ast}(\boldsymbol{\theta}) 
    \right\},
\]
where $\widehat{\theta}^{\ast}(\boldsymbol{\lambda})$ denotes the NP-GMM estimator evaluated at a given value of $\boldsymbol{\lambda}$, defined as:
\[
    \widehat{\theta}^{\ast}(\boldsymbol{\lambda}) 
    \equiv 
    \arg \min_{\boldsymbol{\theta}} \;
    \widehat{Q}(\boldsymbol{\lambda}, \boldsymbol{\theta}).
\]
If $\widehat{\Upsilon}$ contains a single element, that value is the NP-GMM estimator. In the case of multiple fixed points, the NP-GMM estimator is defined as the element in $\widehat{\Upsilon}$ that minimizes the criterion function $\widehat{Q}$.

Let $\boldsymbol{\theta}_{0}$ be the true value the structural parameters in the population, and let $\boldsymbol{\lambda}_{0}$ be the corresponding value for $\boldsymbol{\lambda}$, i.e., $\boldsymbol{\lambda}_{0} = \lambda^{\ast}(\boldsymbol{\theta}_{0})$. Remember that the mapping
$\lambda^{\ast}(\boldsymbol{\theta})$ is deterministic, i.e., it does not incorporate sampling/estimation error. 

Define the population counterpart of the sample criterion function $\widehat{Q} \left( \boldsymbol{\lambda}, \boldsymbol{\theta} \right)$:
\begin{equation}
    Q_{0} \left( \boldsymbol{\lambda}, \boldsymbol{\theta} \right) 
    \; \equiv \;
    \mathbb{E} 
    \left[ \,
        \widehat{Q} \left( \boldsymbol{\lambda}, \boldsymbol{\theta} \right)
    \, \right],
\end{equation}
the corresponding mapping $\theta_{0}^{\ast}(\boldsymbol{\lambda})$,
\begin{equation}
    \theta_{0}^{\ast}(\boldsymbol{\lambda})
    \equiv 
    \arg \min_{\boldsymbol{\theta}} \;
    Q_{0} \left( \boldsymbol{\lambda}, \boldsymbol{\theta} \right),
\end{equation}
and the set of population NP-GMM fixed points:
\begin{equation}
    \Upsilon_{0} \equiv 
    \left\{ 
    (\boldsymbol{\lambda}, \boldsymbol{\theta}) 
    \in [0, 1]^N \times \Theta : \, 
    \boldsymbol{\theta} = \theta_{0}^{\ast}(\boldsymbol{\lambda}) 
    \text{ and } 
    \boldsymbol{\lambda} = \lambda^{\ast}(\boldsymbol{\theta}) 
    \right\},
\end{equation}
It is important to note that Assumption \ref{assumption_1} does not imply that $(\boldsymbol{\lambda}_0, \boldsymbol{\theta}_0)$  is the only point in the set $\Upsilon_{0}$. Our identification assumption does not rule that $\Upsilon_{0}$ may have multiple elements. But, as we show below, this does not affect the consistency of the NP-GMM estimator. 

\subsection{Consistency of the NP-GMM Estimator}

We begin with an outline of the proof, which proceeds in four steps:
\begin{itemize}
    \item Step 1. The true vector $(\boldsymbol{\lambda}_0, \boldsymbol{\theta}_0)$ uniquely minimizes 
    $Q_0(\boldsymbol{\lambda}, \boldsymbol{\theta})$ within the set $\Upsilon_0$.
    
    \item Step 2. With probability approaching 1, every element of $\widehat{\Upsilon}$ belongs to an arbitrarily small open ball around an element of $\Upsilon_0$.

    \item Step 3. With probability approaching 1, the NP-GMM estimator belongs to an open ball around the true $(\boldsymbol{\lambda}_0, \boldsymbol{\theta}_0)$.
\end{itemize}

If $\Upsilon_0$ is a singleton, then Steps 1 and 2 prove the consistency of the NP-GMM estimator. Otherwise, if $\Upsilon_0$ contains multiple NP-GMM fixed points, then Steps 1 to 3 prove the consistency of the NP-GMM estimator.

\subsubsection*{Step 1 — $(\boldsymbol{\lambda}_0, \boldsymbol{\theta}_0)$ uniquely minimizes 
$Q_0(\boldsymbol{\lambda}, \boldsymbol{\theta})$ within the set $\Upsilon_0$}

\begin{itemize}

\item [1.a.] This is a direct implication of identification Assumption \ref{assumption_1}. 

\item [1.b.] First, we can establish that $(\boldsymbol{\lambda}_0, \boldsymbol{\theta}_0)$ belongs to the set $\Upsilon_0$. By construction, $\boldsymbol{\lambda}_0 = \lambda^{\ast}(\boldsymbol{\theta}_0)$, and the correct specification of the model --embedded into Assumption \ref{assumption_1} -- implies that $Q_{0}(\boldsymbol{\lambda}_0, \boldsymbol{\theta})$ is minimized at $\boldsymbol{\theta} = \boldsymbol{\theta}_0$.  Therefore, $(\boldsymbol{\lambda}_0, \boldsymbol{\theta}_0)$ belongs to $\Upsilon_0$.

\item [1.c.]  Second, by Assumption \ref{assumption_1}, for any element $(\boldsymbol{\lambda}, \boldsymbol{\theta}) \in \Upsilon_0$ different to $(\boldsymbol{\lambda}_0, \boldsymbol{\theta}_0)$, we have that $Q_{0}(\boldsymbol{\lambda}, \boldsymbol{\theta}) \, > \, 
Q_{0}(\boldsymbol{\lambda}_0, \boldsymbol{\theta}_0) = 0$.

\end{itemize}

%%%%%%

\subsubsection*{Step 2 — Convergence in probability of $\widehat{\Upsilon}$ to $\Upsilon_0$}

\begin{itemize}
\item [2.a.] Let $(\widehat{\boldsymbol{\lambda}}, \widehat{\boldsymbol{\theta}})$ be an element of $\widehat{\Upsilon}$. For each element of $\Upsilon_0$, consider an arbitrarily small open ball that contains it. Let $\mathcal{O}$ be the union of these open balls. We want to prove that as the sample size increases, we have that $\Pr\left((\widehat{\boldsymbol{\lambda}}, \widehat{\boldsymbol{\theta}}) \in \mathcal{O}\right) \to 1$, i.e., every element in $\widehat{\Upsilon}$ belongs to an arbitrarily small ball around an element in $\Upsilon_0$.

\item [2.b.] Define the function
\begin{equation}
    \Delta_0(\boldsymbol{\lambda}, \boldsymbol{\theta}) 
    \equiv 
    Q_{0}(\boldsymbol{\lambda}, \boldsymbol{\theta}) -
    \min_{c \in \Theta}
    \left\{ 
        Q_{0}(c, \boldsymbol{\lambda})
    \right\}.
\end{equation}
Because $Q_{0}(\boldsymbol{\lambda}, \boldsymbol{\theta})$ is continuous and $[0, 1]^{N} \times \Theta$ is compact, Berge's maximum theorem establishes that $\Delta_0(\boldsymbol{\lambda}, \boldsymbol{\theta})$ is a continuous function on $[0, 1]^{N} \times \Theta$. By construction, $\Delta_0(\boldsymbol{\lambda}, \boldsymbol{\theta}) \geq 0$ for all $(\boldsymbol{\lambda}, \boldsymbol{\theta}) \in [0, 1]^{N} \times \Theta$.

\item [2.c.] Given the set of arbitrarily small balls $\mathcal{O}$, we can construct a scalar $\varepsilon > 0$ that provides a distance between $\mathcal{O}$ and $\Upsilon_0$. We now describe how we construct $\varepsilon$. 
    \begin{itemize}
        \item [2.c.1.] Let $\mathcal{E}$ be the set of vectors $(\boldsymbol{\lambda}, \boldsymbol{\theta})$ that satisfy the model restrictions $\boldsymbol{\lambda} = \lambda^{\ast}(\boldsymbol{\theta})$:
        \begin{equation}
            \mathcal{E} \equiv 
            \left\{ 
                (\boldsymbol{\lambda}, \boldsymbol{\theta}) \in [0, 1]^{N} \times \Theta : \,
                \boldsymbol{\lambda} = \lambda^{\ast}(\boldsymbol{\theta}) 
        \right\}.
        \end{equation}
        Since $[0, 1]^{N} \times \Theta$ is compact and the function $\boldsymbol{\lambda} - \lambda^{\ast}(\boldsymbol{\theta})$ is continuous, the set $\mathcal{E}$ is compact. By definition, $\Upsilon_0 \subset \mathcal{E}$. Because both $\mathcal{E}$ and $\mathcal{O}^c$ are compact, the set $\mathcal{O}^c \cap \mathcal{E}$ is compact. 

        \item [2.c.2.] Define the constant:
        \begin{equation}
            \varepsilon = 
            \min_{(\boldsymbol{\lambda}, \boldsymbol{\theta}) 
            \in \mathcal{O}^c \cap \mathcal{E}} \; 
            \Delta_0(\boldsymbol{\lambda}, \boldsymbol{\theta}).
        \end{equation}
        By construction, $\varepsilon > 0$. To see this, suppose $\varepsilon = 0$. Then there exists $(\boldsymbol{\lambda}, \boldsymbol{\theta})$ which does not belong to $\Upsilon_0$ but belongs to $\mathcal{E}$ and minimizes $Q_{0}^{\ast}(\boldsymbol{\lambda}, \boldsymbol{\theta})$ -- as zero is the minimum value this function can reach-- which contradicts the definition of $\Upsilon_0$.
    \end{itemize}

    \item [2.d.] Since $\Theta$ is a compact set and $\widehat{Q}(\boldsymbol{\lambda}, \boldsymbol{\theta})$ is continuous and bounded, we have that $\widehat{Q}(\boldsymbol{\lambda}, \boldsymbol{\theta})$ converges uniformly to $Q_0(\boldsymbol{\lambda}, \boldsymbol{\theta})$. Define the following indicator of a sample event:
    \begin{equation}
        \widehat{\mathcal{A}} \equiv 
        1\left\{ 
            \left| 
                \widehat{Q}(\boldsymbol{\lambda}, \boldsymbol{\theta}) - 
                Q_0(\boldsymbol{\lambda}, \boldsymbol{\theta})
            \right| < 
            \frac{\varepsilon}{2} 
            \text{ for all } 
            (\boldsymbol{\lambda}, \boldsymbol{\theta}) 
            \in \Theta \times [0,1]^{N} 
        \right\}.
    \end{equation}
    Uniform convergence of $\widehat{Q}$ to $Q_{0}$ implies that, as the sample size increases, $Pr(\widehat{\mathcal{A}}=1)$ approaches one for any arbitrary $\varepsilon > 0$.

    \item [2.e.] Let $(\widehat{\boldsymbol{\lambda}}, \widehat{\boldsymbol{\theta}})$ be an element of $\widehat{\Upsilon}$ and let the sample size be large enough such that $\widehat{\mathcal{A}}=1$. Then:
\begin{itemize}
    \item[i.] Since $\widehat{\mathcal{A}}=1$, we have that 
    $Q_0(\widehat{\boldsymbol{\lambda}}, \widehat{\boldsymbol{\theta}}) -
    \widehat{Q}(\widehat{\boldsymbol{\lambda}}, \widehat{\boldsymbol{\theta}}) < \varepsilon/2$.
    
    \item[ii.] Since $\widehat{\mathcal{A}}=1$, for any $\boldsymbol{\theta} \in \Theta$, we have that 
    $\widehat{Q}(\widehat{\boldsymbol{\lambda}}, \boldsymbol{\theta}) -
    Q_0(\widehat{\boldsymbol{\lambda}}, \boldsymbol{\theta})
     < \varepsilon/2$.
    
    \item[iii.] Since $(\widehat{\boldsymbol{\lambda}}, \widehat{\boldsymbol{\theta}})$ is an NP-GMM fixed point, for any $\boldsymbol{\theta} \in \Theta$, $\widehat{Q}(\widehat{\boldsymbol{\lambda}}, \widehat{\boldsymbol{\theta}}) \leq \widehat{Q}(\widehat{\boldsymbol{\lambda}}, \boldsymbol{\theta})$.
\end{itemize}

Combining $(i)$ and $(iii)$ gives:
\begin{equation}
    Q_0(\widehat{\boldsymbol{\lambda}}, \widehat{\boldsymbol{\theta}}) < 
    \widehat{Q}(\widehat{\boldsymbol{\lambda}}, \boldsymbol{\theta}) 
    + \frac{\varepsilon}{2}, \quad \forall \boldsymbol{\theta} \in \Theta.
\end{equation}
Adding $(ii)$, we obtain:
\begin{equation}
    Q_0(\widehat{\boldsymbol{\lambda}}, \widehat{\boldsymbol{\theta}}) < 
    Q_0(\widehat{\boldsymbol{\lambda}}, \boldsymbol{\theta}) 
    + \varepsilon, \quad \forall \boldsymbol{\theta} \in \Theta.
\end{equation}
Hence,
\begin{equation}
    \Delta_0(\widehat{\boldsymbol{\lambda}}, \widehat{\boldsymbol{\theta}}) \equiv
    Q_0(\widehat{\boldsymbol{\lambda}}, \widehat{\boldsymbol{\theta}}) -
    \min_{\boldsymbol{\theta} \in \Theta} 
    Q_0(\widehat{\boldsymbol{\lambda}}, \boldsymbol{\theta}) 
    < \varepsilon.
\end{equation}
Since $(\widehat{\boldsymbol{\lambda}}, \widehat{\boldsymbol{\theta}}) \in \mathcal{E}$ and $\varepsilon = \min_{(\boldsymbol{\lambda}, \boldsymbol{\theta}) \in \mathcal{O}^c \cap \mathcal{E}} \; 
\Delta_0(\boldsymbol{\lambda}, \boldsymbol{\theta})$, this implies that $(\widehat{\boldsymbol{\lambda}}, \widehat{\boldsymbol{\theta}}) \in \mathcal{O}$.

Therefore, as the sample size increases and $Pr(\widehat{\mathcal{A}}=1)$ approaches one, we have that $\Pr\left((\widehat{\boldsymbol{\lambda}}, \widehat{\boldsymbol{\theta}}) \in \mathcal{O}\right) \to 1$, i.e., every element in $\widehat{\Upsilon}$ belongs to an arbitrarily small ball around an element in $\Upsilon_0$.

\item [2.f.] As a corollary, note that if $\Upsilon_0$ is a singleton, consistency follows immediately.

\end{itemize}

%%%%%%

\subsubsection*{Step 3 — The NP-GMM estimator belongs to an open ball around $(\boldsymbol{\lambda}_0, \boldsymbol{\theta}_0)$}

\begin{itemize}
    \item [3.a] Let $\mathcal{O}_0$ be an open ball around the true $(\boldsymbol{\lambda}_0, \boldsymbol{\theta}_0)$, and let $\mathcal{O}_1$ be the union of open balls around all other elements in $\Upsilon_0 \setminus \{(\boldsymbol{\lambda}_0, \boldsymbol{\theta}_0)\}$. Define the constant $\eta$:
    \begin{equation}
        \eta  = 
        \inf_{(\boldsymbol{\lambda}, \boldsymbol{\theta}) 
        \in \mathcal{O}_1} \, 
        Q_0(\boldsymbol{\lambda}, \boldsymbol{\theta}) -
        \sup_{(\boldsymbol{\lambda}, \boldsymbol{\theta}) 
        \in \mathcal{O}_0} \, 
        Q_0(\boldsymbol{\lambda}, \boldsymbol{\theta}).
    \end{equation}
    By continuity of $Q_0$ and since $(\boldsymbol{\lambda}_0, \boldsymbol{\theta}_0) $ is isolated, we can take $\mathcal{O}_0$ and $\mathcal{O}_1$ small enough to ensure $\eta > 0$.

    \item [3.b] Let $(\widehat{\boldsymbol{\lambda}}, \widehat{\boldsymbol{\theta}})$ be the NP-GMM fixed point in $\mathcal{O}_0$. Let $\widehat{\mathcal{A}}$ be the event indicator defined in Step 2. And let $\varepsilon^{\ast} > 0$ be the minimum between the constants $\eta$ defined above and the $\varepsilon$ defined in Step 2. Then:
    \begin{itemize}
        \item [i.] Since $\widehat{\mathcal{A}}=1$, $\widehat{Q}(\widehat{\boldsymbol{\lambda}}, \widehat{\boldsymbol{\theta}})
        < Q_{0}(\widehat{\boldsymbol{\lambda}}, \widehat{\boldsymbol{\theta}}) + \varepsilon^{\ast}/2$.
    
        \item [ii.] Since $\widehat{\mathcal{A}}=1$, for any $(\boldsymbol{\lambda}, \boldsymbol{\theta}) \in \Theta \times [0,1]^{N}$, $Q_{0}(\boldsymbol{\lambda}, \boldsymbol{\theta}) 
        < \widehat{Q}(\boldsymbol{\lambda}, \boldsymbol{\theta}) + \varepsilon^{\ast}/2$
    
        \item[iii.] Since $(\widehat{\boldsymbol{\lambda}}, \widehat{\boldsymbol{\theta}}) \in \mathcal{O}_0$, for any $(\boldsymbol{\lambda}, \boldsymbol{\theta}) \in \widehat{\Upsilon} \cap \mathcal{O}_1$, 
        $Q_0(\widehat{\boldsymbol{\lambda}}, \widehat{\boldsymbol{\theta}}) \leq Q_0(\boldsymbol{\lambda}, \boldsymbol{\theta}) - \varepsilon^{*}$.
    \end{itemize}
  
    Combining (i) and (iii),
    \begin{equation}
        \widehat{Q}(\widehat{\boldsymbol{\lambda}}, \widehat{\boldsymbol{\theta}}) < 
        Q_0(\boldsymbol{\lambda}, \boldsymbol{\theta}) -
        \varepsilon^{*}/2 \quad 
        \forall (\boldsymbol{\lambda}, \boldsymbol{\theta}) 
        \in \widehat{\Upsilon} \cap \mathcal{O}_1.
    \end{equation}
    Adding (ii), we obtain,
        \begin{equation}
        \widehat{Q}(\widehat{\boldsymbol{\lambda}}, \widehat{\boldsymbol{\theta}}) < 
        Q_0(\boldsymbol{\lambda}, \boldsymbol{\theta}) -
        \quad 
        \forall (\boldsymbol{\lambda}, \boldsymbol{\theta}) 
        \in \widehat{\Upsilon} \cap \mathcal{O}_1.
    \end{equation}
    By Step 2, $\widehat{\Upsilon} \rightarrow_{p} \Upsilon_{0}$ and this implies that with probability approaching one 
    $\widehat{\Upsilon} \subset \mathcal{O}_0 \cup \mathcal{O}_1$. Therefore, with probability approaching one, the NP-GMM estimator lies in $\mathcal{O}_0$ and hence converges in probability to the true $(\boldsymbol{\lambda}_0, \boldsymbol{\theta}_0)$. 
    
    Q.E.D.
\end{itemize}

\subsection{Asymptotic distribution of the NP-GMM estimator}

The marginal conditions that define the NP-GMM estimator are:

\begin{equation}
    \left\{
    \begin{array}{rcl}
        \displaystyle \frac{1}{J} 
        \sum_{j} \, g_{j}(\widehat{\boldsymbol{\lambda}}, 
        \widehat{\boldsymbol{\theta}}) & = & 0
        \\
        \widehat{\boldsymbol{\lambda}} - \lambda^{\ast} 
        (\widehat{\boldsymbol{\theta}}) & = & 0
    \end{array}
    \right.
\end{equation}
Applying a stochastic mean value expansion between $(\boldsymbol{\lambda}_{0}, \boldsymbol{\theta}_{0})$ and $(\widehat{\boldsymbol{\lambda}}, \widehat{\boldsymbol{\theta}})$, and using the consistency of $(\widehat{\boldsymbol{\lambda}}, \widehat{\boldsymbol{\theta}})$, we have:
\begin{equation}
    \left\{
    \begin{array}{rcl}
        \displaystyle
        \frac{1}{\sqrt{J}} 
        \sum_{j} \frac{\partial g_{j}^{0}}{\partial \boldsymbol{\theta}} -
        \Omega_{\theta \theta} \, \sqrt{J}
        (\widehat{\boldsymbol{\theta}} - \boldsymbol{\theta}_0) -
        \Omega_{\theta \lambda} \, \sqrt{J}
        (\widehat{\boldsymbol{\lambda}} - \boldsymbol{\lambda}_0) 
        & = & o_p(\sqrt{J}) \\
        \sqrt{J}(\widehat{\boldsymbol{\lambda}} - \boldsymbol{\lambda}_0) -
        \Lambda_{\theta} \, 
        \sqrt{J}
        (\widehat{\boldsymbol{\theta}} - \boldsymbol{\theta}_0)
        & = & o_p(\sqrt{J}) 
    \end{array}
    \right.
\end{equation}
with $\Omega_{\theta \theta} \equiv \displaystyle \mathbb{E} \left[ \frac{\partial^{2} g_{j}(\boldsymbol{\lambda}_0, \boldsymbol{\theta}_0)}{\partial \boldsymbol{\theta} \partial \boldsymbol{\theta}^{\prime}}\right]$, $\Omega_{\theta \lambda} \equiv \displaystyle \mathbb{E} \left[ \frac{\partial^{2} g_{j}(\boldsymbol{\lambda}_0, \boldsymbol{\theta}_0)}{\partial \boldsymbol{\theta} \partial \boldsymbol{\lambda}^{\prime}}\right]$, and $\Lambda_{\theta} \equiv \displaystyle \frac{\partial \boldsymbol{\lambda}^{\ast}(\boldsymbol{\theta}_0)}{\partial \boldsymbol{\theta}^{\prime}}$. 
Solving the second equation for $\sqrt{J}(\widehat{\boldsymbol{\lambda}} - \boldsymbol{\lambda}_0)$ and substituting into the first equation gives:
\begin{equation}
    \left[ 
        \Omega_{\theta \theta} + \Omega_{\theta \lambda} \, \Lambda_{\theta}
    \right] \, 
    \sqrt{J} (\widehat{\boldsymbol{\theta}} - \boldsymbol{\theta}_0)
    \; = \; 
    \displaystyle
    \frac{1}{\sqrt{J}} 
    \sum_{j} \frac{\partial g_{j}^{0}}{\partial \boldsymbol{\theta}} + 
    o_p(\sqrt{J}) 
\end{equation}
Thus, by the Mann–Wald theorem:
\[
    \sqrt{J} (\widehat{\boldsymbol{\theta}} - \boldsymbol{\theta}_0)
    \overset{d}{\longrightarrow} 
    \mathcal{N}(0, V_{np}),
\]
where
\begin{equation}
    V_{np} \; = \; 
    \left[ 
        \Omega_{\theta \theta} + \Omega_{\theta \lambda} \, \Lambda_{\theta}
    \right]^{-1} 
    \, \Omega_{\theta \theta} \,
    \left[ 
        \Omega_{\theta \theta} + \Lambda_{\theta}^{\prime} 
        \, \Omega_{\theta \lambda}^{\prime}
    \right]^{-1} 
\end{equation}

Q.E.D.

\newpage

\section{Multithreading}
\label{omega faster}

We use MATLAB's profiler to identify the source of computational cost reductions when using NP-GMM. Tables \ref{tab: more eva time 500 25}, \ref{tab: more eva time 50 25} and \ref{tab: more eva time 100 100} compare ABLP and NP-GMM for 1,000 evaluations of the objective function and associated gradients, using different numbers of threads and under different model settings (all with $N = 1,000$).

NP-GMM is faster than ABLP in evaluating both the GMM objective function and its gradients. For objective-function evaluation, as discussed above, NP-GMM’s speed advantage arises from its use of an ``almost'' linear mapping in $\boldsymbol{\sigma}$ to recover $\boldsymbol{\delta}$. For gradient evaluation, NP-GMM further benefits from simple, closed-form analytic gradients that are particularly well suited to parallel computation and multithreading.

In contrast, ABLP relies on more complex numerical procedures to compute $\boldsymbol{\delta}$ and its derivatives. In particular, its linear approximation requires a matrix inversion that jointly accounts for all products within a market, leading to higher-dimensional matrix operations. As shown in the three tables, when each method is implemented using its respective optimal number of threads, NP-GMM evaluates the objective function and gradients approximately five times faster than ABLP under these model settings.

\begin{table}[H]\centering
\caption{Evaluation Time (seconds): ABLP versus NP-GMM ($J = 25$, $T = 500$)}
\label{tab: more eva time 500 25}
\resizebox{.60\columnwidth}{!}{%
\begin{tabular}{ccccc} \hline \hline
Number of Threads                                     &  Method    & GMM ($\boldsymbol{\delta}$)  & Gradient (Jacobian)    & Total \\ \hline
\multirow{2}{*}{1}  & ABLP & 420.6(420.2) & 1285.8(1278.2) & 1706.4 \\
                    &                      NP-GMM  & 251.2(250.8)  & 320 & 571.2 \\ \hline 
\multirow{2}{*}{2}  & ABLP & 285.5(285)	& 1346.2(1338.7) &	1631.7
 \\
                    &                      NP-GMM   & 148.1(147.6)	  & 261.4  & 	409.5
 \\ \hline 
\multirow{2}{*}{3}  & ABLP & 247.8(247.3)	 & 1390.3(1382.7) &	1638.1
 \\
                    &                     NP-GMM    & 114.9(114.4) &	248.7	& 363.6
 \\ \hline 
\multirow{2}{*}{4}  & ABLP & 232.9(232.4)	 & 1674(1666.3) &	1906.9
 \\
                    &                      NP-GMM  & 99.1(98.6)	& 247.2 &	346.3 
 \\ \hline 
\multirow{2}{*}{5}  & ABLP & 225.5(224.9) &	1748.1(1740.4) &	1973.6
 \\
                    &                      NP-GMM    & 91.6(91.1) & 	249.7 &	341.3
 \\ \hline 
\multirow{2}{*}{6}  & ABLP & 230.2(229.6) &	2240.3(2232.7) &	2470.5
 \\
                    &                     NP-GMM    & 87.7(87.2) & 	251.5 & 	339.2
 \\ \hline \hline
\end{tabular}
}
\vspace{0.5em}
\begin{minipage}{0.6\columnwidth}
\footnotesize
\textit{Note:} Entries report mean wall-clock time (in seconds) for a single evaluation of the GMM objective function and its gradient. In the GMM column, numbers in parentheses report the time required to recover the mean utility vector $\boldsymbol{\delta}$. In the Gradient column, numbers in parentheses report the time spent computing the Jacobian associated with the differentiation of the mean utility vector $\boldsymbol{\delta}$. Total time is the sum of GMM and gradient evaluation times.
\end{minipage}
\end{table}

\begin{table}[H]\centering
\caption{Evaluation Time (seconds): ABLP versus NP-GMM ($J = 25$, $T = 50$)}
\label{tab: more eva time 50 25}
\resizebox{.6\columnwidth}{!}{%
\begin{tabular}{ccccc} \hline \hline
Number of Threads                                     &  Method    & GMM ($\boldsymbol{\delta}$)  & Gradient (Jacobian)    & Total \\ \hline
\multirow{2}{*}{1}  & ABLP & 45(44.8) &	127.4(126.6) & 172.4
 \\
                    &                      NP-GMM  & 26.6(26.4) &	 31.8 &	58.4
 \\ \hline 
\multirow{2}{*}{2}  & ABLP & 29.7(29.6) & 	131.1(130.2) &	160.8

 \\
                    &                      NP-GMM   & 15.6(15.4)	& 25.7	& 41.3

 \\ \hline 
\multirow{2}{*}{3}  & ABLP & 25.4(25.3)	& 142.6(141.7)	& 168

 \\
                    &                     NP-GMM    & 12.1(11.9)	& 24.3	& 36.4

 \\ \hline 
\multirow{2}{*}{4}  & ABLP & 23.6(23.4)	& 157.7(156.8) &	181.3

 \\
                    &                      NP-GMM  & 11.1(10.9)	& 23.8 &	 34.9
 
 \\ \hline 
\multirow{2}{*}{5}  & ABLP & 23(22.9)	& 159.1(158.1) &	182.1

 \\
                    &                      NP-GMM    & 10.4(10.2) & 	23.7	 & 34.1

 \\ \hline 
\multirow{2}{*}{6}  & ABLP & 22.8(22.7)	& 172.9(171.8) &	195.7

 \\
                    &                     NP-GMM    & 10.1(9.9) &	23.8 &	33.9

 \\ \hline \hline
\end{tabular}
}
\end{table}

\begin{table}[H]\centering
\caption{Evaluation Time (seconds): ABLP versus NP-GMM ($J = 100$, $T = 100$)}
\label{tab: more eva time 100 100}
\resizebox{.6\columnwidth}{!}{%
\begin{tabular}{ccccc} \hline \hline
Number of Threads                                     &  Method    & GMM ($\boldsymbol{\delta}$)  & Gradient (Jacobian)    & Total \\ \hline
\multirow{2}{*}{1}  & ABLP & 365(364.6) &	 1173.7(1166.9) & 1538.7
 \\
                    &                      NP-GMM  & 211.7(211.3) &	 268.6 &	480.3
 \\ \hline 
\multirow{2}{*}{2}  & ABLP & 253.3(252.8) & 	1082.1(1075.2) &	1335.4

 \\
                    &                      NP-GMM   & 120.9(120.4)	& 208.9	& 329.8

 \\ \hline 
\multirow{2}{*}{3}  & ABLP & 217(216.5)	& 1116.2(1109.4)	& 1333.2

 \\
                    &                     NP-GMM    & 92.9(92.4)	& 205.1	& 298

 \\ \hline 
\multirow{2}{*}{4}  & ABLP & 208.3(207.8)	& 1195.2(1188.2) &	1403.5

 \\
                    &                      NP-GMM  & 83.2(82.7)	& 204.8 &	 288
 
 \\ \hline 
\multirow{2}{*}{5}  & ABLP & 213.3(212.8)	& 1286.2(1279.4) &	1499.5

 \\
                    &                      NP-GMM    & 76.6(76.1) & 	207.6	 & 284.2

 \\ \hline 
\multirow{2}{*}{6}  & ABLP & 217.4(216.8)	& 1364(1356.7) &	1581.4

 \\
                    &                     NP-GMM    & 72.4(71.9) &	207.8 &	280.2

 \\ \hline \hline
\end{tabular}
}
\end{table}

From the discussion above, we observe the speed advantage of NP-GMM over ABLP under the benchmark large sample setting of (a) $J = 25$, $T = 500$, $N = 1,000$. To assess the effect of increasing sample size on NP-GMM's speed advantage, we modify setting (a) by doubling $J$, $T$, and $N$ to create the following settings: (b) $J = 50$, $T = 500$, $N = 1,000$; (c) $J = 25$, $T = 1,000$, $N = 1,000$; and (d) $J = 25$, $T = 500$, $N = 2000$. Figure \ref{fig:eva speed} provides wall-clock time comparisons under these four settings. The optimal number of threads for ABLP is always 2, while the optimal number of threads for NP-GMM is either 5 or 6, with negligible differences between the two. Once again, our algorithm is approximately five times faster than ABLP in evaluating the objective function and associated gradients.

%Remark: I can't run the code for the model setting of $J = 25$ and $T = 2000$. The error code is "Requested 50000x50000 (18.6GB) array exceeds maximum array size preference (15.8GB). This might cause MATLAB to become unresponsive." 

\begin{figure}[H]
     \centering
     \caption{Speed Comparisons under Different Model Settings}
        \label{fig:eva speed}
     \begin{subfigure}[b]{0.48\textwidth}
         \centering
         \includegraphics[width=\textwidth]{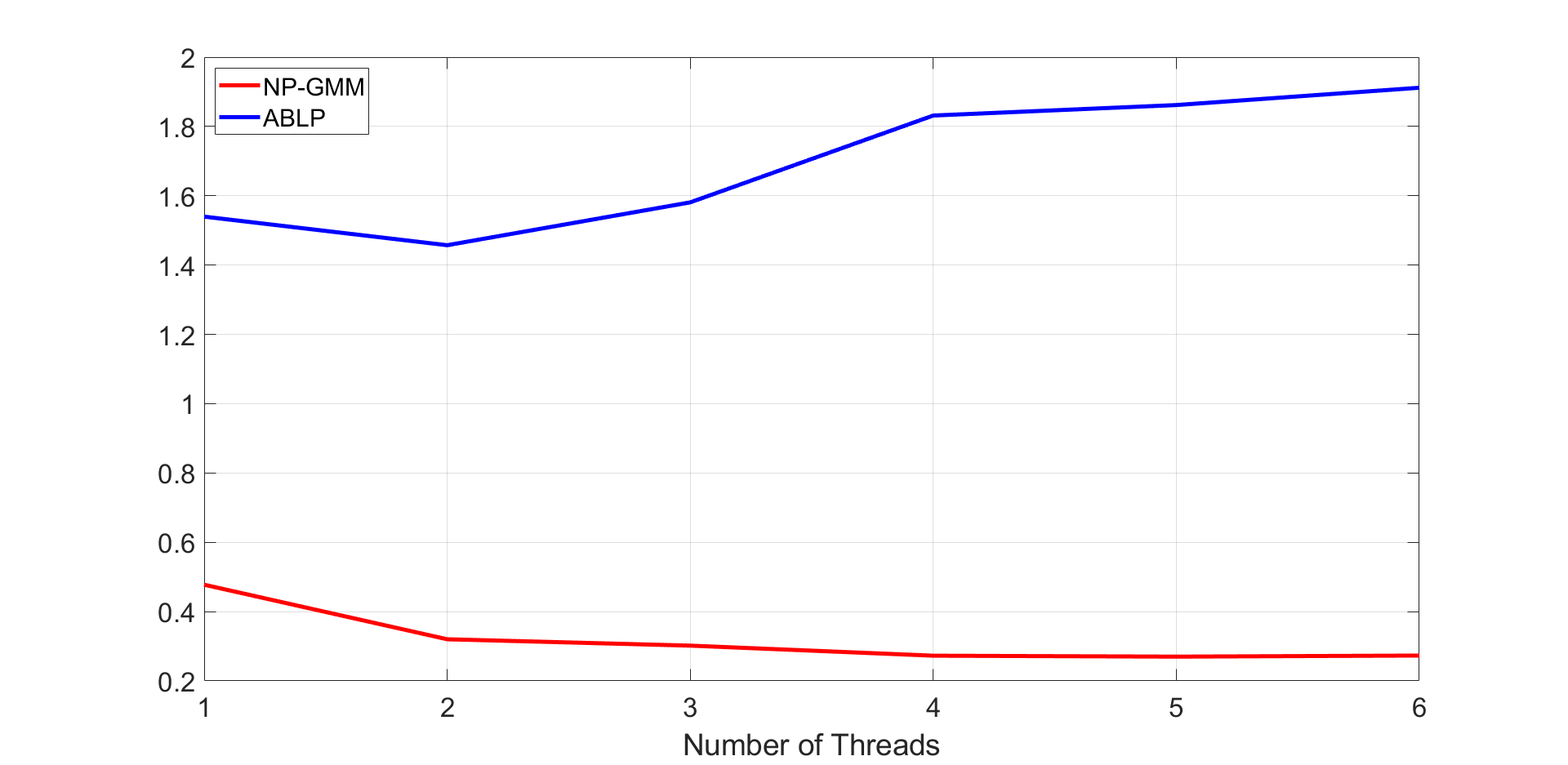}
         \caption{$J = 25$, $T = 500$, $N = 1,000$}
     \end{subfigure}
     %\hfill
     \begin{subfigure}[b]{0.48\textwidth}
         \centering
         \includegraphics[width=\textwidth]{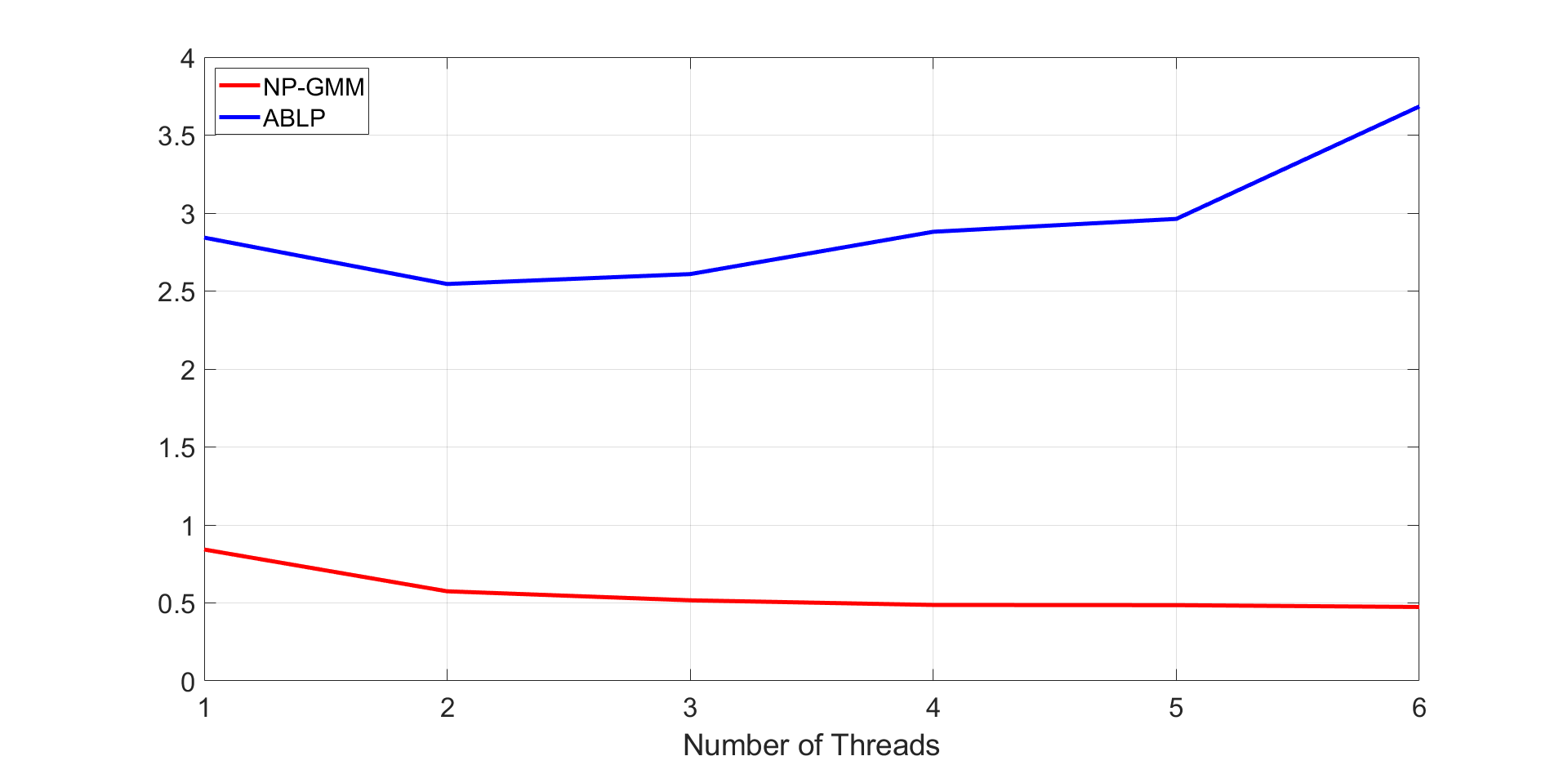}
         \caption{$J = 50$, $T = 500$, $N = 1,000$}
     \end{subfigure}
     \vfill
     \begin{subfigure}[b]{0.48\textwidth}
         \centering
         \includegraphics[width=\textwidth]{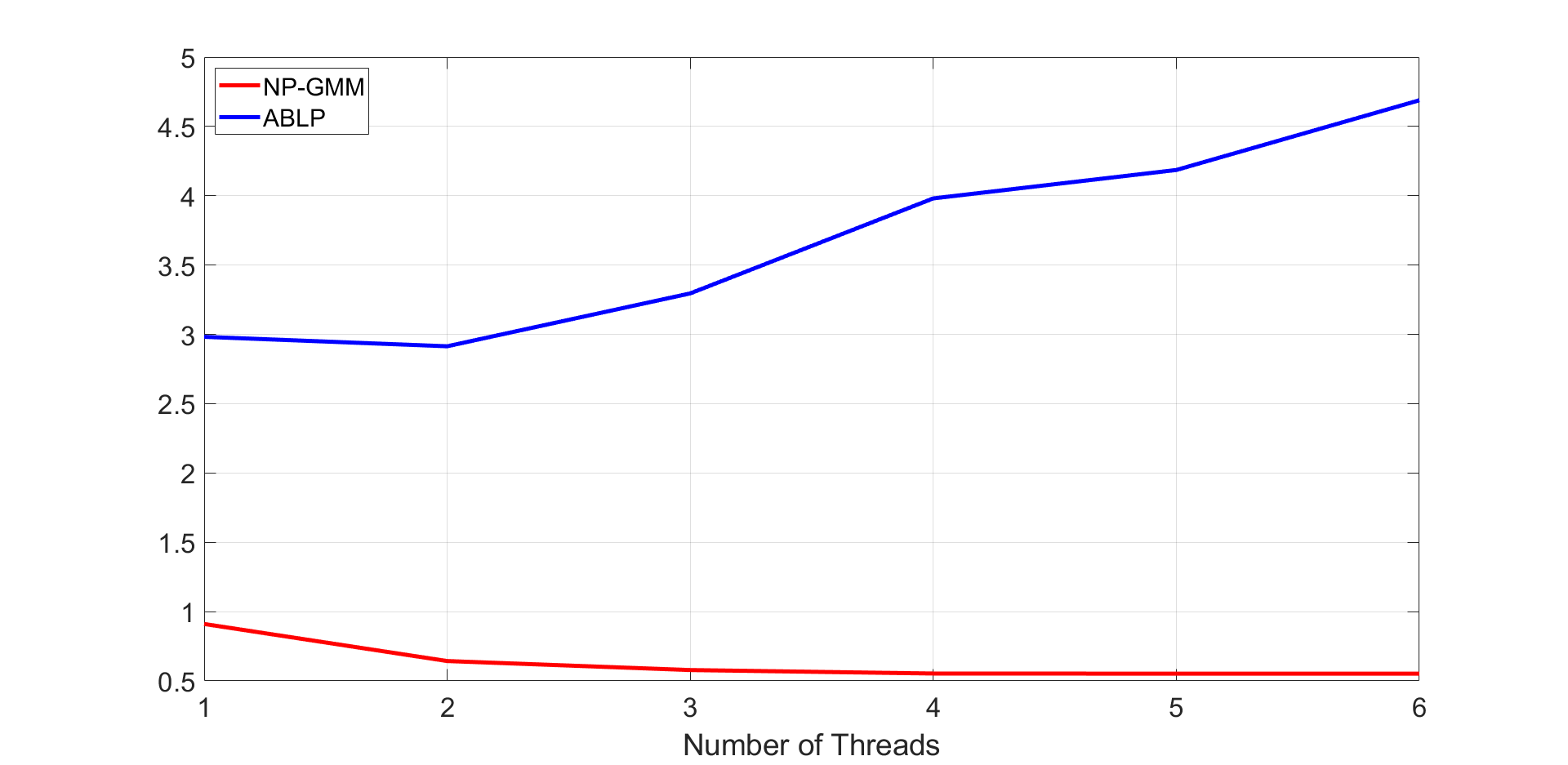}
         \caption{$J = 25$, $T = 1,000$, $N = 1,000$}
     \end{subfigure}
     %\hfill
     \begin{subfigure}[b]{0.48\textwidth}
         \centering
         \includegraphics[width=\textwidth]{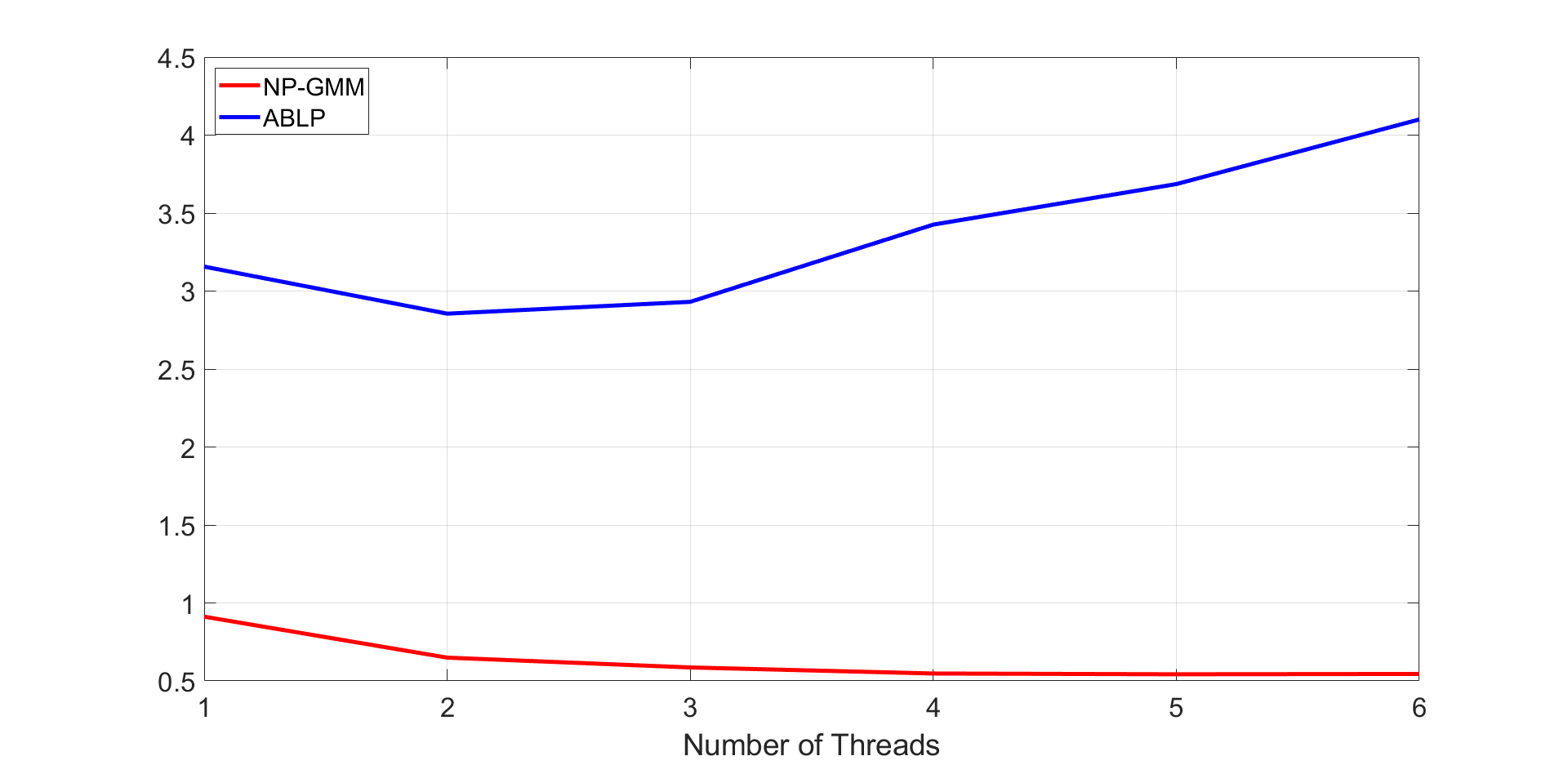}
         \caption{$J = 25$, $T = 500$, $N = 2000$}
     \end{subfigure}
\end{figure}

We conduct more tests by fixing one dimension of the model settings ($J$ or $T$) and increasing the size of the other dimension. The results are consistent with our previous finding that NP-GMM benefits from parallel computing and multithreading, whereas ABLP is less suited to these techniques, as shown in Figure \ref{fig:more eva speed}. We also calculate the speed gains of NP-GMM and ABLP using the optimal number of threads compared to a single thread, as shown in Tables \ref{tab:more eva time J 25} and \ref{tab:more eva time T 500}. NP-GMM reduces wall-clock time by 40\% to 47\% with the optimal number of threads, while ABLP achieves savings of only 2\% to 14\%.

%The time ratio of optimal thread over single thread ranges from 52.91\% to 60.64\% for NP-GMM, and 86.35\% to 97.72\% for ABLP.

\begin{figure}[H]
     \centering
     \caption{More Speed Comparisons}
        \label{fig:more eva speed}
     \begin{subfigure}[b]{0.48\textwidth}
         \centering
         \includegraphics[width=\textwidth]{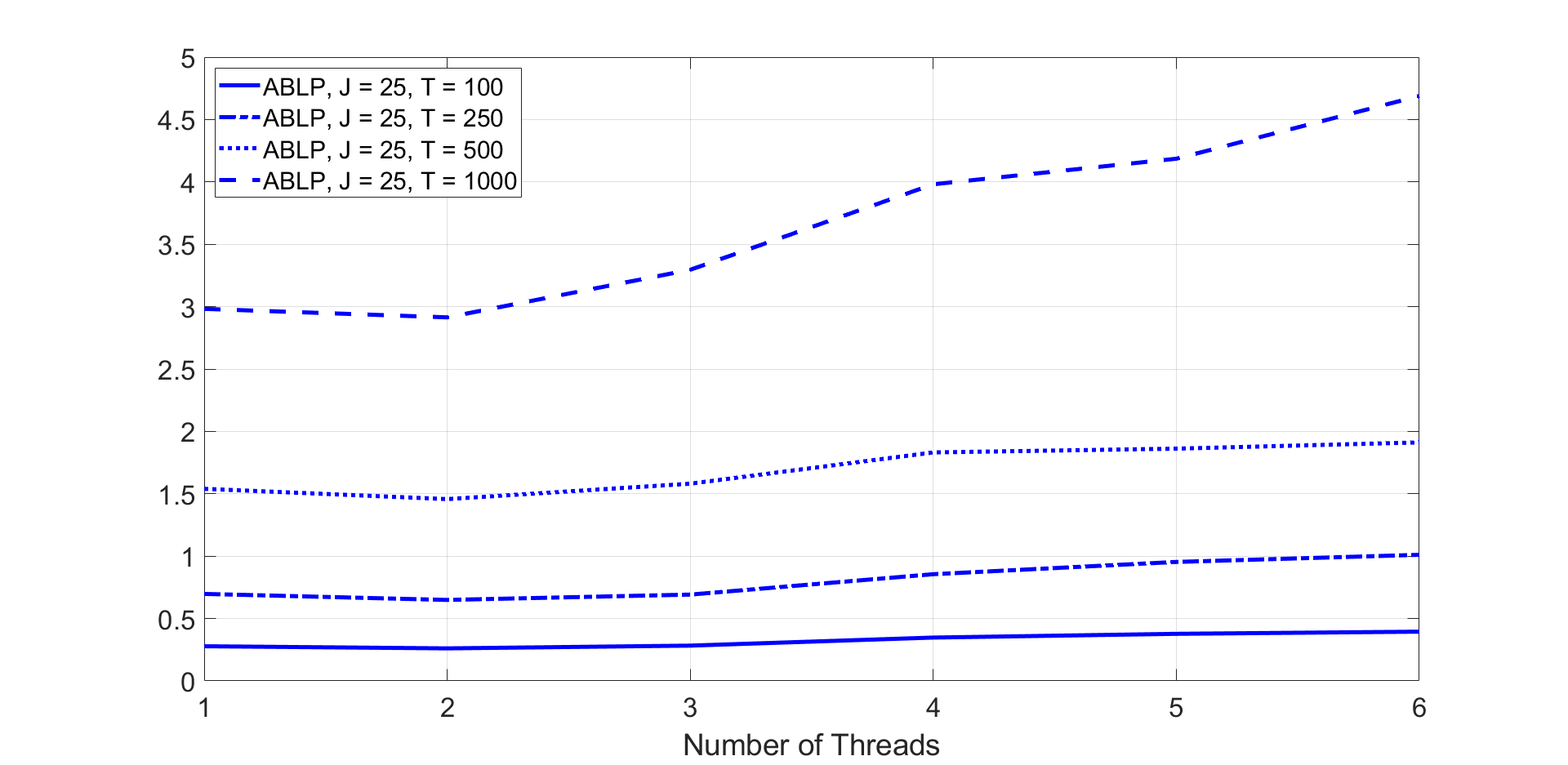}
         \caption{ABLP, $J = 25$, $N = 1,000$}
     \end{subfigure}
     %\hfill
     \begin{subfigure}[b]{0.48\textwidth}
         \centering
         \includegraphics[width=\textwidth]{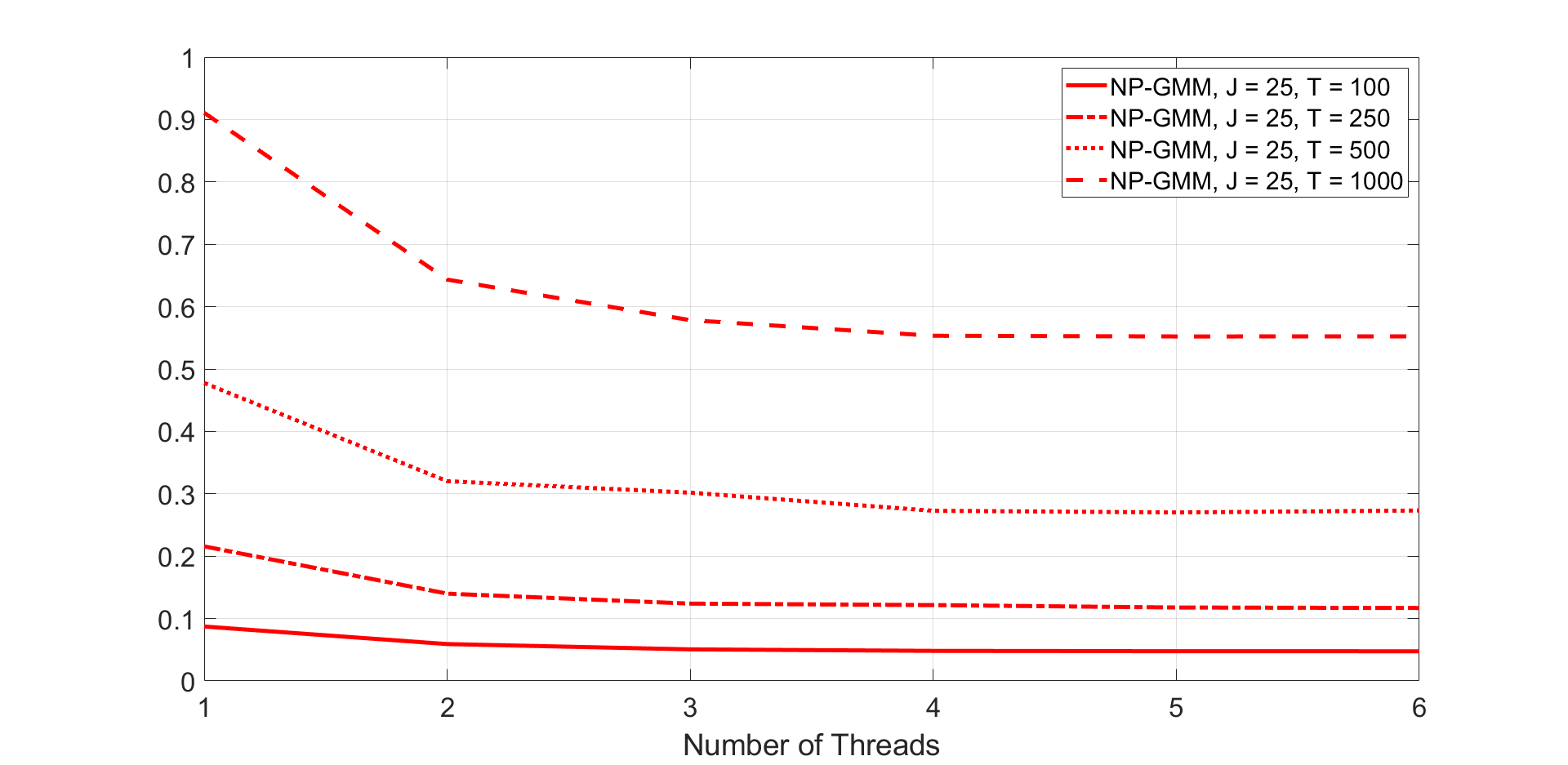}
         \caption{NP-GMM, $J = 25$, $N = 1,000$}
     \end{subfigure}
     \vfill
     \begin{subfigure}[b]{0.48\textwidth}
         \centering
         \includegraphics[width=\textwidth]{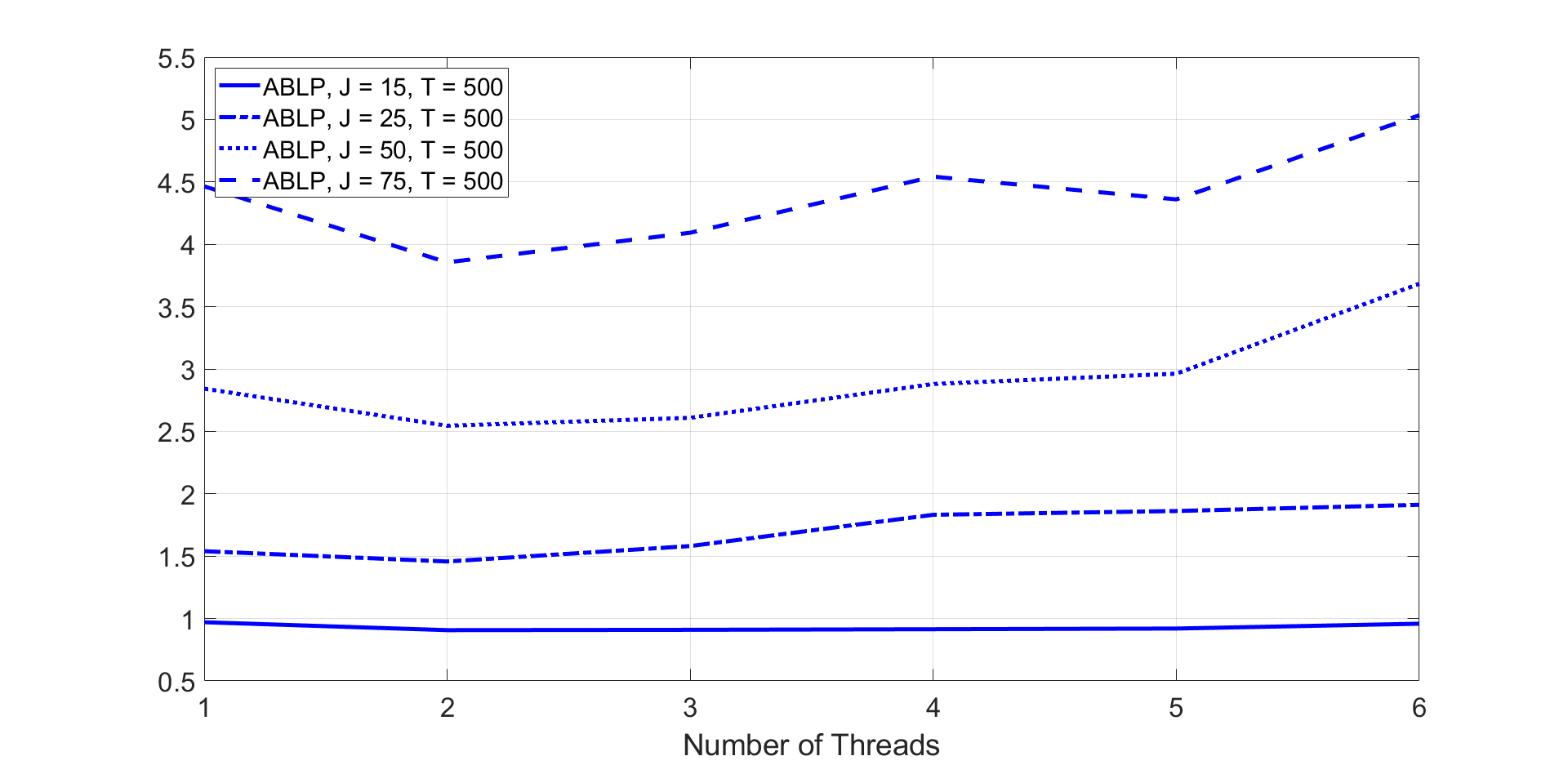}
         \caption{ABLP, $T = 500$, $N = 1,000$}
     \end{subfigure}
     %\hfill
     \begin{subfigure}[b]{0.48\textwidth}
         \centering
         \includegraphics[width=\textwidth]{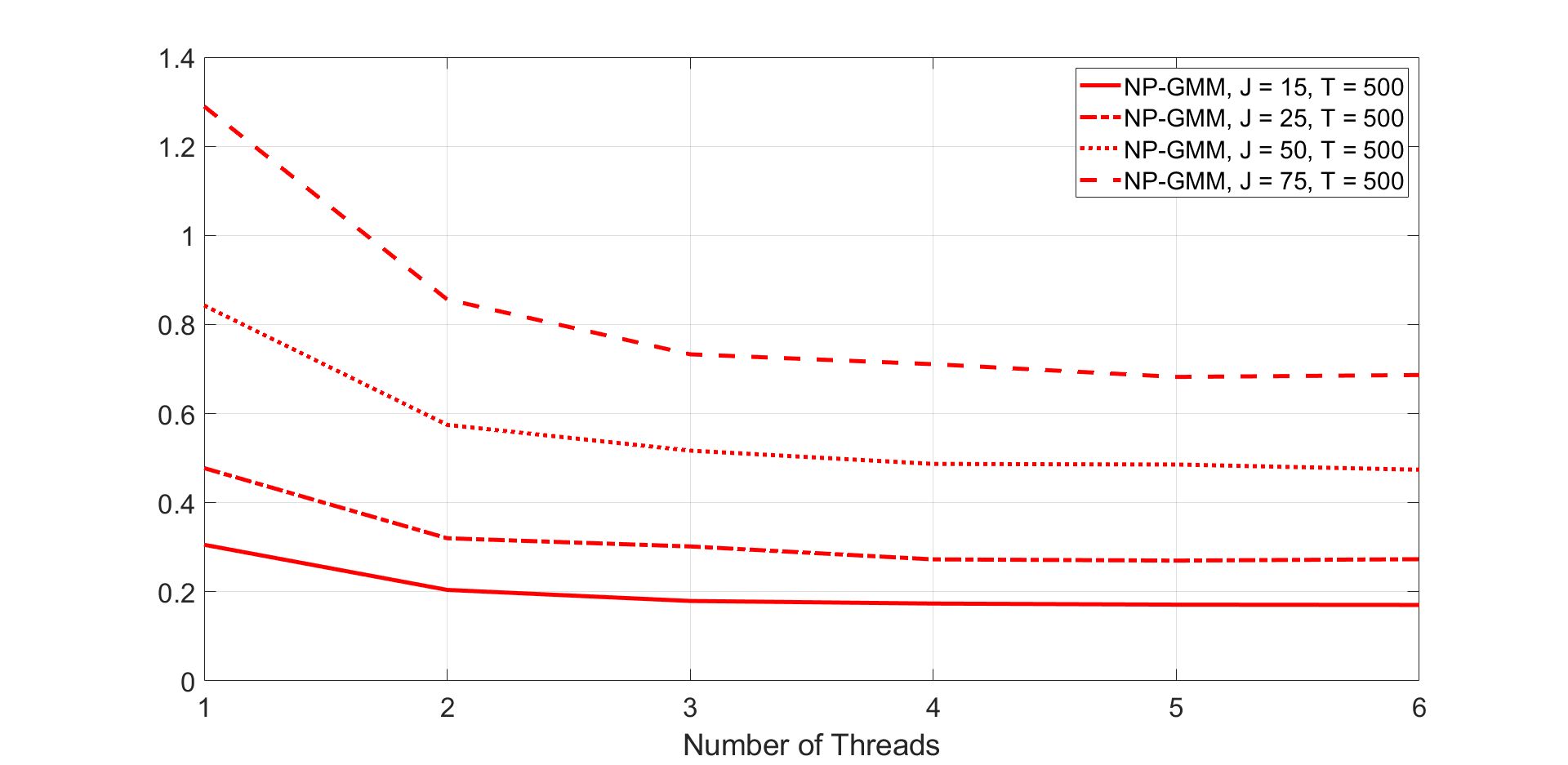}
         \caption{NP-GMM, $T = 500$, $N = 1,000$}
     \end{subfigure}
\end{figure}

\begin{table}[H]\centering
\caption{Evaluation Time (seconds): ABLP versus NP-GMM ($J = 25$)}
\label{tab:more eva time J 25}
\resizebox{.9\columnwidth}{!}{%
\begin{tabular}{cccccc} \hline \hline
$T$ &      & Optimal number of threads & Optimal time    & Single thread time & Optimal/Single thread \\ \hline
\multirow{2}{*}{100}  & ABLP & 2 &	 0.2608 & 0.2779  &  93.84\%
 \\
                    &                      NP-GMM  & 6 &	 0.0475  & 0.0872  &  54.47\%
 \\ \hline 
\multirow{2}{*}{250}  & ABLP & 2 & 	0.6491 &	0.6974  &  93.07\%

 \\
                    &                      NP-GMM   & 6 	& 0.1169 	& 0.2157  & 54.2\%

 \\ \hline 
\multirow{2}{*}{500}  & ABLP & 2	& 1.4573 	& 1.5398  & 94.64\%

 \\
                    &                     NP-GMM    & 5	& 0.2700 	& 0.4776  & 56.53\%

 \\ \hline 
\multirow{2}{*}{1000}  & ABLP & 2	& 2.9141 &	 2.9822  & 97.72\%

 \\
                    &                      NP-GMM  & 5	& 0.5523 &	 0.9108  & 60.64\%

 \\ \hline \hline
\end{tabular}
}
\end{table}

\begin{table}[H]\centering
\caption{Evaluation Time (seconds): ABLP versus NP-GMM ($T = 500$)}
\label{tab:more eva time T 500}
\resizebox{.9\columnwidth}{!}{%
\begin{tabular}{cccccc} \hline \hline
$J$ &      & Optimal number of threads & Optimal time    & Single thread time & Optimal/Single thread \\ \hline
\multirow{2}{*}{15}  & ABLP & 2 &	 0.9069 & 0.9704  &  93.46\%
 \\
                    &                      NP-GMM  & 6 &	 0.1705  & 0.3056  &  55.79\%
 \\ \hline 
\multirow{2}{*}{25}  & ABLP & 2 & 	1.4573 &	1.5398  &  94.64\%

 \\
                    &                      NP-GMM   & 5	& 0.2700	& 0.4776  & 56.53\%

 \\ \hline 
\multirow{2}{*}{50}  & ABLP & 2	& 2.5453 	& 2.8428  &  89.53\%

 \\
                    &                     NP-GMM    & 6	& 0.4740 	& 0.8427  &  56.25\%

 \\ \hline 
\multirow{2}{*}{75}  & ABLP & 2	& 3.8550 &	4.4645  & 86.35\%

 \\
                    &                      NP-GMM  & 5	& 0.6823 &	 1.2895  &  52.91\%

 \\ \hline \hline
\end{tabular}
}
\end{table}

\section{Analytic Gradients of NP-GMM Estimators} \label{anaytic gradient}

In this section, we derive the analytic gradient of the NP-GMM objective function. The objective function and its derivatives are explicit. For convenience, we concentrate out the linear parameters $\boldsymbol{\beta}$ and use $\boldsymbol{\theta}$ to denote the remaining parameters in this section. Let $X$ and $Z$ denote the stacked matrices of regressors and instruments, respectively, formed by stacking all $\boldsymbol{x_{jt}}$ and $\boldsymbol{z_{jt}}$ across products and markets.

The GMM objective function is:
\begin{equation}
\begin{aligned}
    & & & 
    Q \big(\boldsymbol{\lambda}, \boldsymbol{\theta} \big) 
    \; = \; 
    \boldsymbol{\xi}(\boldsymbol{\lambda}, \boldsymbol{\theta})^{\prime} \; Z \; W \; Z^{\prime} \;
    \boldsymbol{\xi}(\boldsymbol{\lambda}, \boldsymbol{\theta}), \\
& \text{where}  & & \boldsymbol{\theta} \; = \; (\boldsymbol{\sigma}, \boldsymbol{\pi})', \\ 
    & & & \boldsymbol{\xi}(\boldsymbol{\lambda}, \boldsymbol{\theta}) \; = \; \boldsymbol{\delta}(\boldsymbol{\lambda},\boldsymbol{\theta}) - X \; \boldsymbol{\beta}(\boldsymbol{\lambda}, \boldsymbol{\theta}), \\ 
    & & & \delta_{jt}(\boldsymbol{\lambda}, \boldsymbol{\theta}) 
    \; = \;
    \ln s_{jt} - \ln \Big[\frac{1}{N} \sum_{i=1}^{N} \lambda_{it} 
    \exp \big(\sum_{k=1}^{K} \boldsymbol{x}_{jt}^k(\sigma^{k} v_i^k +\sum_{r=1}^{R} \pi_{kr} d_{ir} ) \big) \Big], \\
    & & & \boldsymbol{\beta}(\boldsymbol{\lambda}, \boldsymbol{\theta}) 
    \; = \; 
    \Big[ 
        (X^\prime \; Z \; W \; Z^\prime X)^{-1} 
        X^\prime \; Z \; W \; Z^\prime
    \Big] \;
    \boldsymbol{\delta}(\boldsymbol{\lambda}, \boldsymbol{\theta}).
\end{aligned}
\end{equation}
Thus,
\begin{equation}
    \boldsymbol{\xi}(\boldsymbol{\lambda}, \boldsymbol{\theta}) 
    \; = \; 
    \boldsymbol{\delta}(\boldsymbol{\lambda}, \boldsymbol{\theta}) - 
    X \; 
    \Big[ 
        (X^\prime \; Z \; W \; Z^\prime X)^{-1} 
        X^\prime \; Z \; W \; Z^\prime
    \Big] \;
    \boldsymbol{\delta}(\boldsymbol{\lambda}, \boldsymbol{\theta}).
\end{equation}
The gradient is:
\begin{equation}
\begin{aligned}
    & & & \nabla_{\theta}Q(\boldsymbol{\lambda}, \boldsymbol{\theta}) 
    \; = \;  
    \frac{d\boldsymbol{\xi}(\boldsymbol{\lambda}, \boldsymbol{\theta})^{\prime}}{d\boldsymbol{\theta}} \; \frac{dQ}{d\boldsymbol{\xi}} 
    \; = \; 
    2 \; \frac{d\boldsymbol{\xi}(\boldsymbol{\lambda}, \boldsymbol{\theta})^{\prime}}{d\boldsymbol{\theta}} 
    \; Z \; W \; Z^{\prime} \; 
    \boldsymbol{\xi}(\boldsymbol{\lambda}, \boldsymbol{\theta}), \\
& \text{where}
    & & \frac{d\boldsymbol{\xi}(\boldsymbol{\lambda}, \boldsymbol{\theta})}{d\boldsymbol{\theta}} \; = \;  
    \frac{d\boldsymbol{\delta}(\boldsymbol{\lambda}, \boldsymbol{\theta})}{d\boldsymbol{\theta}} - 
    X \; 
    \Big[ 
        (X^\prime \; Z \; W \; Z^\prime X)^{-1} 
        X^\prime \; Z \; W \; Z^\prime
    \Big] \;
    \frac{d\boldsymbol{\delta}(\boldsymbol{\lambda}, \boldsymbol{\theta})}{d\boldsymbol{\theta}}; \\
    & \text{and}
    & & \frac{\partial\delta_{jt}(\boldsymbol{\lambda}, \boldsymbol{\theta})}{\partial \sigma^{k}} = - \frac{\sum_{i=1}^{N} \lambda_{it}
    \exp \big(\sum_{k=1}^{K} \boldsymbol{x}_{jt}^k(\sigma^{k} v_i^k +\sum_{r=1}^{R} \pi_{kr} d_{ir} ) \big) \boldsymbol{x}_{jt}^k v_i^k}{ \sum_{i=1}^{N} \lambda_{it} 
    \exp \big(\sum_{k=1}^{K} \boldsymbol{x}_{jt}^k(\sigma^{k} v_i^k +\sum_{r=1}^{R} \pi_{kr} d_{ir} ) \big)}, \\
    & & & \frac{\partial\delta_{jt}(\boldsymbol{\lambda}, \boldsymbol{\theta})}{\partial \pi_{kr}} = - \frac{\sum_{i=1}^{N} \lambda_{it}
    \exp \big(\sum_{k=1}^{K} \boldsymbol{x}_{jt}^k(\sigma^{k} v_i^k +\sum_{r=1}^{R} \pi_{kr} d_{ir} ) \big) \boldsymbol{x}_{jt}^k d_{ir}}{\sum_{i=1}^{N} \lambda_{it} 
    \exp \big(\sum_{k=1}^{K} \boldsymbol{x}_{jt}^k(\sigma^{k} v_i^k +\sum_{r=1}^{R} \pi_{kr} d_{ir} ) \big)}.
\end{aligned}
\end{equation}
Both numerator and denominator of the above two equations are closed-form and "almost" linear functions of $\boldsymbol{\theta}$, so their calculation is fast.

\section{Supplementary Results for the LCBO Application}
 \label{estimate_details_lcbo} 

\begin{table}[ht]
\centering
\caption{Mean Wall-clock Time per Evaluation by Number of Threads}
\label{tab:per_eval_time_lcbo}
\begin{tabular}{l l c c c c c c}
\hline \hline
& & \multicolumn{6}{c}{Number of Threads} \\
\cmidrule(lr){3-8}
$J$ & Method & 1 & 2 & 3 & 4 & 5 & 6 \\
\hline
\multirow{2}{*}{800}
& ABLP   & 27.04 & 19.22 & 16.82 & 15.91 & 15.79 & 15.95 \\
& NP-GMM &  2.28 &  1.54 &  1.35 &  1.20 &  1.24 &  1.17 \\
\hline
\multirow{2}{*}{1{,}600}
& ABLP   & 100.12 & 68.37 & 57.34 & 56.04 & 52.68 & 52.49 \\
& NP-GMM &   4.58 &  3.12 &  2.67 &  2.51 &  2.43 &  2.39 \\
\hline
\multirow{2}{*}{2{,}900}
& ABLP   & 366.92 & 254.19 & 233.02 & 195.90 & 202.34 & 192.40 \\
& NP-GMM &   8.47 &   5.66 &   4.91 &   4.56 &   4.34 &   4.30 \\
\hline \hline
\end{tabular}

\vspace{0.5em}
\begin{minipage}{0.9\linewidth}
\footnotesize
\textit{Note:} Entries report wall-clock time (in seconds) for a single evaluation of the objective
function and its gradients. Each model uses $T = 76$ markets and $N = 1{,}000$ simulated consumers.
Each entry is the mean over 10 runs.
\end{minipage}
\end{table}

\clearpage

\end{appendices}

\bibliographystyle{econometrica}
\bibliography{references}

\end{document}